%% file: paper.tex
\documentclass[11pt]{article}

\usepackage{times}
\usepackage{wrapfig}
\usepackage{citesort}
\usepackage{graphicx}
\usepackage{color}
\usepackage{fullpage}
\usepackage{amssymb}
\usepackage{amsmath}
\usepackage{color}
\usepackage{url}

\newcommand{\eps}{\epsilon}
\newcommand{\E}{\mathbf{E}}
\renewcommand{\Pr}{\mathbf{Pr}}
\newcommand{\abs}[1]{\left| #1 \right|}
\newcommand{\norm}[1]{\left\lVert #1 \right\rVert}

\newtheorem{theorem}{Theorem}
\newtheorem{lemma}{Lemma}
\newtheorem{claim}{Claim}
\newtheorem{property}{Property}
\newtheorem{fact}{Fact}

\newtheorem{corollary}{Corollary}

\newtheorem{remark}{Remark}
\newtheorem{definition}{Definition}

\newcommand{\poly}{\text{poly}}
\newcommand{\nnz}{\text{nnz}}

\newcommand{\qinhides}[1]{}

\newenvironment{proof}{\trivlist\item[]\emph{Proof}:}%
{\unskip\nobreak\hskip 1em plus 1fil\nobreak$\Box$
\parfillskip=0pt%
\endtrivlist}

\begin{document}

\title{Subspace Embeddings and $\ell_p$-Regression \\ Using Exponential Random Variables}

\author{David P. Woodruff\\IBM Almaden\\dpwoodru@us.ibm.com 
\and Qin Zhang\\IBM Almaden\\qinzhang@cse.ust.hk}
  
%\date{}

\maketitle

\begin{abstract}
Oblivious low-distortion subspace embeddings are a crucial building block for numerical linear algebra problems. 
We show for any real $p, 1 \leq p < \infty$, given a matrix $M \in \mathbb{R}^{n \times d}$ with $n \gg d$, with constant probability we can choose a matrix $\Pi$ with $\max(1, n^{1-2/p}) \poly(d)$ rows and $n$ columns so that simultaneously for all $x \in \mathbb{R}^d$,
$\|Mx\|_p \leq \|\Pi Mx\|_{\infty} \leq \poly(d) \|Mx\|_p.$
Importantly, $\Pi M$ can be computed in the optimal $O(\nnz(M))$ time, where $\nnz(M)$ is the number of non-zero entries of $M$. This generalizes all previous oblivious subspace embeddings which required $p \in [1,2]$ due to their use of $p$-stable random variables. Using our matrices $\Pi$, we also improve the best known distortion of oblivious subspace embeddings of $\ell_1$ into $\ell_1$ with $\tilde{O}(d)$ target dimension in $O(\nnz(M))$ time from $\tilde{O}(d^3)$ to $\tilde{O}(d^2)$, which can further be improved to $\tilde{O}(d^{3/2}) \log^{1/2} n$ if $d = \Omega(\log n)$, answering a question of Meng and Mahoney (STOC, 2013).

We apply our results to $\ell_p$-regression, obtaining a $(1+\eps)$-approximation in $O(\nnz(M)\log n) + \poly(d/\eps)$ time, improving the best known $\poly(d/\eps)$ factors for every $p \in [1, \infty) \setminus \{2\}$. If one is just interested in a $\poly(d)$ rather than a $(1+\eps)$-approximation to $\ell_p$-regression, a corollary of our results is that for all $p \in [1, \infty)$ we can solve the $\ell_p$-regression problem without using general convex programming, that is, since our subspace embeds into $\ell_{\infty}$ it suffices to solve a linear programming problem.  Finally, we give the first protocols for the distributed $\ell_p$-regression problem for every $p \geq 1$ which are nearly optimal in communication and computation. 
\end{abstract}

\input{intro}

\input{preliminary}

%\input{algo}

%\input{analysis}

\input{big-p}

\input{small-p}

\input{regression}

\input{communication}
%\section{Regression in the Distributed Setting}
%Due to space constraints, we leave this section to Appendix~\ref{sec:dist-regression}.

%\bibliographystyle{abbrv}
%\bibliography{paper}
%

%\newpage

%\appendix
%\input{appendix}

\end{document}

%% file: intro.tex
%\qinsays{I have replaced all $A$ with $M$. In this paper $A$ is reserved to denote the Auerbach matrix.}

\section{Introduction}
An oblivious subspace embedding with distortion $\kappa$ is a distribution over linear maps $\Pi: \mathbb{R}^n \rightarrow \mathbb{R}^t$ for which for any fixed $d$-dimensional subspace of $\mathbb{R}^n$, represented as the column space of an $n \times d$ matrix $M$, with constant probability, $\|Mx\|_p \leq \|\Pi Mx\|_p \leq \kappa \|Mx\|_p$ simultaneously for all vectors $x \in \mathbb{R}^d$. The goal is to minimize $t$, $\kappa$, and the time to compute $\Pi \cdot M$. For a vector $v$, $\|v\|_p = (\sum_{i=1}^{n} |v_i|^p )^{1/p}$ is its $p$-norm. 

Oblivious subspace embeddings have proven to be an essential ingredient for quickly and approximately solving numerical linear algebra problems. One of the canonical problems is regression, which is well-studied in the learning community, see \cite{cd11,g11,hkz12,kst12} for some recent advances. S\'arlos \cite{s06} observed that oblivious subspace embeddings could be used to approximately solve least squares regression and low rank approximation, and he used fast Johnson-Lindenstrauss transforms \cite{AL08,AC09} to obtain the fastest known algorithms for these problems at the time. Optimizations to this in the streaming model are in \cite{cw09,KN12}.

As an example, in least squares regression, one is given an $n \times d$ matrix $M$ which is usually overconstrained, i.e., $n \gg d$, as well as a vector $b \in \mathbb{R}^n$. The goal is to output $x^* = \textrm{argmin}_x \|Mx-b\|_2$, that is, to find the vector $x^*$ so that $Mx^*$ is the (Euclidean) projection of $b$ onto the column space of $M$. This can be solved exactly in $O(nd^2)$ time. Using fast Johnson-Lindenstrauss transforms, S\'arlos was able to find a vector $x'$ with $\|Mx'-b\|_2 \leq (1+\eps)\|Mx^*-b\|_2$ in $O(nd \log d) + \poly(d/\eps)$ time, providing a substantial improvement. The application of oblivious subspace embeddings (to the space spanned by the columns of $M$ together with $b$) is immediate: given $M$ and $b$, compute $\Pi M$ and $\Pi b$, and solve the problem $\min_x \|\Pi Mx-\Pi b\|_2$. If $\kappa = (1+\eps)$ and $t \ll n$, one obtains a relative error approximation by solving a much smaller instance of regression. 

Another line of work studied $\ell_p$-regression for $p \neq 2$. One is given an $n \times d$ matrix $M$ and an $n \times 1$ vector $b$, and one seeks $x^* =  \textrm{argmin}_x \|Mx-b\|_p$. For $1 \leq p < 2$, this provides a more robust form of regression than least-squares, since the solution is less sensitive to outliers. For $2 < p \leq \infty$, this is even more sensitive to outliers, and can be used to remove outliers. While $\ell_p$-regression can be solved in $\poly(n)$ time for every $1 \leq p \leq \infty$ using convex programming, this is not very satisfying if $n \gg d$. For $p = 1$ and $p = \infty$ one can use linear programming to solve these problems, though for $p = 1$ the complexity will still be superlinear in $n$. Clarkson \cite{Clarkson05} was the first to achieve an $n \cdot \poly(d)$ time algorithm for $\ell_1$-regression, which was then extended to $\ell_p$-regression for every $1 \leq p \leq \infty$ with the same running time \cite{DDHKM09}. 

The bottleneck of these algorithms for $\ell_p$-regression was a preprocessing step, in which one well-conditions the matrix $M$ by choosing a different basis for its column space. Sohler and Woodruff \cite{SW11} got around this for the important case of $p = 1$ by designing an oblivious subspace embedding $\Pi$ for which $\|Mx\|_1 \leq \|\Pi Mx\|_1  = O(d \log d) \|Mx\|_1$ in which $\Pi$ has $O(d \log d)$ rows. Here, $\Pi$ was chosen to be a matrix of Cauchy random variables. Instead of running the expensive conditioning step on $M$, it is run on $\Pi M$, which is much smaller. One obtains a $d \times d$ change of basis matrix $R^{-1}$. Then one can show the matrix $\Pi MR^{-1}$ is well-conditioned. This reduced the running time for $\ell_1$-regression to $nd^{\omega -1} + \poly(d/\eps)$, where $\omega < 3$ is the exponent of matrix multiplication. The dominant term is the $nd^{\omega -1}$, which is the cost of computing $\Pi M$ when $\Pi$ is a matrix of Cauchy random variables. 

In \cite{CDMMMW12}, Clarkson et. al combined the ideas of Cauchy random variables and Fast Johnson Lindenstrauss transforms to obtain a more structured family of subspace embeddings, referred to as the FCT1 in their paper, thereby improving the running time for $\ell_1$-regression to $O(nd \log n) + \poly(d/\eps)$. An alternate construction, referred to as the FCT2 in their paper, gave a family of subspace embeddings that was obtained by partitioning the matrix $M$ into $n/\poly(d)$ blocks and applying a fast Johnson Lindenstrauss transform on each block. Using this approach, the authors were also able to obtain an $O(nd \log n) + \poly(d/\eps)$ time algorithm for $\ell_p$-regression for every $1 \leq p \leq \infty$. 

While the above results are nearly optimal for dense matrices, one could hope to do better if the number of non-zero entries of $M$, denoted $\nnz(M)$, is much smaller than $nd$. Indeed, $M$ is often a sparse matrix, and one could hope to achieve a running time of $O(\nnz(M)) + \poly(d/\eps)$. Clarkson and Woodruff \cite{CW12} designed a family of sparse oblivious subspace embeddings $\Pi$ with $\poly(d/\eps)$ rows, for which $\|Mx\|_2 \leq \|\Pi Mx\|_2 \leq (1+\eps)\|Mx\|_2$ for all $x$. Importantly, the time to compute $\Pi M$ is only $\nnz(M)$, that is, proportional to the sparsity of the input matrix. The $\poly(d/\eps)$ factors were optimized by Meng and Mahoney \cite{MM12}, Nelson and Nguyen \cite{NN12}, and Miller and Peng \cite{MP12}. Combining this idea with that in the FCT2, they achieved running time $O(\nnz(M) \log n) + \poly(d/\eps)$ for $\ell_p$-regression for any constant $p$, $1 \leq p < \infty$.

Meng and Mahoney \cite{MM12} gave an alternate subspace embedding family to solve the $\ell_p$-regression problem in $O(\nnz(M) \log n) + \poly(d/\eps)$ time for $1 \leq p < 2$. One feature of their construction is that the number of rows in the subspace embedding matrix $\Pi$ is only $\poly(d)$, while that of Clarkson and Woodruff \cite{CW12} for $1 \leq p < 2$ is $n/\poly(d)$. This feature is important in the distributed setting, for which there are multiple machines, each holding a subset of the rows of $M$, who wish to solve an $\ell_p$-regression problem by communicating with a central server. The natural solution is to use shared randomness to agree upon an embedding matrix $\Pi$, then apply $\Pi$ locally to each of their subsets of rows, then add up the sketches using the linearity of $\Pi$. The communication is proportional to the number of rows of $\Pi$. This makes the algorithm of Meng and Mahoney more communication-efficient, since they achieve $\poly(d/\eps)$ communication. However, one drawback of the construction of Meng and Mahoney is that their solution only works for $1 \leq p < 2$. This is inherent since they use $p$-stable random variables, which only exist for $p \leq 2$. 

\subsection{Our Results}
In this paper, we improve all previous low-distortion oblivious subspace embedding results for every $p \in [1,\infty) \backslash \{2\}$. We note that the case $p = 2$ is already resolved in light of \cite{CW12,MM12,NN12}. All results hold with arbitrarily large constant probability. $\gamma$ is an arbitrarily small constant. In all results $\Pi M$ can be computed in $O(\nnz(M))$ time (for the third result, we assume that $\nnz(M) \ge d^{2+\gamma}$). 
%\qinsays{Shall we mention that $p=2$ is solved in your paper with Ken?}
\begin{itemize}
\item A matrix $\Pi \in \mathbb{R}^{O(n^{1 - 2/p} \log n (d \log d)^{1 + 2/p} + d^{5+4p}) \times n}$ for $p > 2$ such that given $M \in \mathbb{R}^{n \times d}$, for $\forall x \in \mathbb{R}^d$, %, with constant probability 
\begin{equation*}
\textstyle
\Omega(1/(d \log d)^{1/p}) \cdot \norm{Mx}_p \le \norm{\Pi M x}_\infty \le O((d \log d)^{1/p}) \cdot \norm{Mx}_p.
\end{equation*}

\item A matrix $\Pi \in \mathbb{R}^{O(d^{1+\gamma}) \times n}$ for $1 \le p < 2$ such that given $M \in \mathbb{R}^{n \times d}$, for $\forall x \in \mathbb{R}^d$, %, with constant probability 
\begin{equation*}
\textstyle \Omega\left(\max\left\{1 / {(d \log d \log n)^{\frac{1}{p}-\frac{1}{2}}}, 1 / (d \log d)^{1/p}\right\}\right) \cdot \norm{Mx}_p \le \norm{\Pi M x}_2 \le O((d \log d)^{1/p}) \cdot \norm{Mx}_p.
\end{equation*}

Note that since $\norm{\Pi M x}_\infty \le \norm{\Pi M x}_2 \le O(d^{(1+\gamma)/2}) \norm{\Pi M x}_\infty$, we can always replace the $2$-norm estimator by the $\infty$-norm estimator with the cost of another $d^{(1+\gamma)/2}$ factor in the distortion.

\item A matrix $\Pi \in \mathbb{R}^{O(d \log^{O(1)} d) \times n}$ such that given $M \in \mathbb{R}^{n \times d}$, for $\forall x \in \mathbb{R}^d$,%, with constant probability 
\begin{equation*}
\Omega\left(\max\left\{1 / (d \log d), 1/\sqrt{d \log d \log n}\right\}\right) \cdot \norm{Mx}_1 \le \norm{\Pi M x}_1 \le O(d \log^{O(1)} d) \cdot \norm{Mx}_1.
\end{equation*}
In \cite{MM12} the authors asked whether a distortion $\tilde{O}(d^3)$~\footnote{
We use $\tilde{O}(f)$ to denote a function of the form $f \cdot \log^{O(1)}(f)$.} is optimal for $p = 1$ for mappings $\Pi M$ that can be computed in $O(\text{nnz}(M))$ time. Our result shows that the distortion can be further improved to $\tilde{O}(d^{2})$, and if one also has $d > \log n$, even further to $\tilde{O}(d^{3/2})\log^{1/2} n$. Our embedding also improves the $\tilde{O}(d^{2+\gamma})$ distortion of the much slower \cite{CDMMMW12}. In Table~\ref{tab:results} we compare our result with previous results for $\ell_1$ oblivious subspace embeddings. Our lower distortion embeddings for $p = 1$ can also be used in place of the $\tilde{O}(d^3)$ distortion embedding of \cite{MM12} in the context of quantile regression \cite{ymm13}. 
\end{itemize}

\begin{table*}[t!]
\centering
\scalebox{0.88}{
\begin{tabular}{|c||c|c|c|}
\hline
 & Time & Distortion & Dimemsion\\
\hline
\cite{SW11}& $nd^{\omega - 1}$ & $\tilde{O}(d)$ & $\tilde{O}(d)$ \\
\hline
\cite{CDMMMW12}& $nd\log d$ & $\tilde{O}(d^{2+\gamma})$ & $\tilde{O}(d^5)$ \\
\hline
\cite{CW12} + \cite{NN12} & $\nnz(A) \log n$ & $\tilde{O}\left(d^{(x+1)/2}\right)\ (x\ge 1)$ & $\tilde{O}(n/d^x)$ \\
\hline
\cite{CW12} + \cite{CDMMMW12} + \cite{NN12} & $\nnz(A) \log n$ & $\tilde{O}(d^3)$ & $\tilde{O}(d)$ \\
\hline
\cite{CW12} + \cite{SW11} + \cite{NN12} & $\nnz(A) \log n$ & $\tilde{O}(d^{1+\omega/2})$ & $\tilde{O}(d)$ \\
\hline
\cite{MM12}& $\nnz(A)$ & $\tilde{O}(d^3)$ & $\tilde{O}(d^5)$ \\
\hline
\cite{MM12} + \cite{NN12} & $\nnz(A) + \tilde{O}(d^6)$ & $\tilde{O}(d^3)$ & $\tilde{O}(d)$ \\
\hline
This paper & $\nnz(A) + \tilde{O}(d^{2})$ & $\tilde{O}(d^2)$ & $\tilde{O}(d)$ \\
           & $\nnz(A) + \tilde{O}(d^{2})$ & $\tilde{O}(d^{3/2}) \log^{1/2} n$
& $\tilde{O}(d)$\\
\hline
\end{tabular}
}
\caption{
Results for $\ell_1$ oblivious subspace embeddings.
%FJLT denotes Fast Johnson-Lindenstrauss Transform; see Lemma~\ref{lem:FJLT} for its properties. 
$\omega < 3$ is the exponent of matrix multiplication.}
%$\gamma$ is an arbitrarily small constant.} 
\label{tab:results}
\end{table*}
Our oblivious subspace embeddings directly lead to improved $(1+\eps)$-approximation results for $\ell_p$-regression for every $p  \in [1,\infty) \backslash \{2\}$. We further implement our algorithms for $\ell_p$-regression in a distributed setting, where we have $k$ machines and a centralized server. The sites want to solve the regression problem via communication. We state both the communication and the time required of our distributed $\ell_p$-regression algorithms. One can view the time complexity of a distributed algorithm as the sum of the time complexities of all sites including the centralized server (see Section~\ref{sec:regression} for details).

Given an $\ell_p$-regression problem specified by $M \in \mathbb{R}^{n \times (d-1)}, b \in \mathbb{R}^n, \eps>0$ and $p$, let $\bar{M} = [M, -b] \in \mathbb{R}^{n \times d}$. Let $\phi(t, d)$ be the time of solving $\ell_p$-regression problem on $t$ vectors in $d$ dimensions. 
\begin{itemize}
\item For $p > 2$, we obtain a distributed algorithm with communication $\tilde{O}\left(k n^{1-2/p} d^{2+2/p} + d^{4+2p}/\eps^2\right)$ and running time 
$\tilde{O}\left(\nnz(\bar{M}) + (k + d^2 ) (n^{1-2/p} d^{2+2/p} + d^{6+4p})+ \phi(\tilde{O}(d^{3+2p}/\eps^2), d)\right)$.
%This is the first result achieving a communication cost whose dependency of $kn$ is sublinear.

\item For $1 \le p < 2$, we obtain a distributed algorithm with communication $\tilde{O}\left(k d^{2+\gamma} + d^5 + d^{3+p}/\eps^2 \right)$ and running time 
$\tilde{O}\left(\text{nnz}(\bar{M}) + k d^{2+\gamma} + d^{7-p/2} + \phi(\tilde{O}(d^{2+p}/\eps^2), d\right)$.

%improving the previous best results in \cite{MM12}~\footnote{There are several $\ell_p$ regression algorithms proposed in \cite{MM12}, all in the same format $O(\text{nnz}(\bar{M}) + f(d) + \phi(g(d), d))$. For simplicity we ignore $\log$ factors and assume $\eps = O(1)$. The best $f(d)$ they got in all algorithms is $d^7$, and the best $g(d)$ they got is $d^{3+p/2}$, while in our result these two terms are $d^{7-p/2}$ and $d^{2+p}$, respectively.}. We also obtain a distributed algorithm with communication cost $O(k d \log n + d^5 \log^2 d)$, where $k$ is the number of machines. This also improves the distributed implementation of the centralized algorithm proposed in \cite{MM12}. 
\end{itemize}
We comment on several advantages of our algorithms over standard iterative methods for solving regression problems. We refer the reader to Section 4.5 of the survey \cite{Mahoney11} for more details. 
\begin{itemize}
\item In our algorithm, there is no assumption on the input matrix $M$, i.e., we do not assume it is well-conditioned. Iterative methods are either much slower than our algorithms if the condition number of $M$ is large, or would result in an additive $\eps$ approximation instead of the relative error $\eps$ approximation that we achieve. 
\item Our work can be used in conjunction with other $\ell_p$-regression algorithms. Namely, since we find a well-conditioned basis, we can run iterative methods on our well-conditioned basis to speed them up. 
\end{itemize}

\subsection{Our Techniques}
%\qinsays{I think the key of this paper is to use the properties of exponentials: first, it has the max-stability property which is good for $p>2$. Second, the inverse of a exponential is dominated by a $p$-stable, which allows us to reuse MM12 for $1 \le p < 2$ no overestimation. Third, the upper tail of exponential decreases exponentially, which is good for no underestimation}

%\qinsays{Another good idea is to use NN12. This greatly simplifies the analysis for $1 \le p < 2$. I think one success of this paper is that we make good use of previous works.}

%\qinsays{For $\ell_p$ regression with $p>2$, we compute a $(\alpha, \beta, \infty)$-well-conditioned matrix which is new.}
%
Meng and Mahoney \cite{MM12} achieve $O(\nnz(M)\log n) + \poly(d)$ time for $\ell_p$-regression with sketches of the form $S \cdot D \cdot M$, where $S$ is a $t \times n$ hashing matrix for $t = \poly(d)$, that is, a matrix for which in each column there is a single randomly positioned entry which is randomly either $1$ or $-1$, and $D$ is a diagonal matrix of $p$-stable random variables. The main issues with using $p$-stable random variables $X$ are that they only exist for $1 \leq p \leq 2$, and that the random variable $|X|^p$ is heavy-tailed {\it in both directions}.
%
%$X$ and $Y$ are independent $p$-stable random variables, then for scalars $a$ and $b$, $aX + aY$ is distributed as $(|a|^p + |b|^p)^{1/p} Z$, where $Z$ is itself a $p$-stable random variable. For $p = 2$, they are Gaussian random variables, while for $p = 1$ they are Cauchy random variables. 
%
%There are several inherent limitations with using $p$-stable random variables. First, $p$-stable random variables exist if and only if $p$ is in the real interval $[1, 2]$, and therefore can only be used for $\ell_p$-regression for $1 \leq p \leq 2$. Second, in analyses of them for regression \cite{SW11,CDMMMW12,MM12}, one needs to bound the concentration of $|X|^p$, the $p$-th power of the absolute value of a $p$-stable random variable. To upper bound the distortion of the embedding, one needs to bound $\Pr[|X|^p > t \cdot \textrm{median}(|X|^p)]$, while to lower bound the distortion one needs to bound $\Pr[|X|^p < 1/t \cdot \textrm{median}(|X|^p)]$, for $t \geq 1$. The smaller these bounds, the smaller the distortion is of the embedding. For $p$-stable random variables for $1 \leq p < 2$, both of these bounds are $\Omega(1/t)$, which means $|X|^p$ is heavy-tailed in both directions. 

We replace the $p$-stable random variable with the reciprocal of an {\it exponential random variable}. Exponential random variables have stability properties with respect to the minimum operation, that is, if $u_1, \ldots, u_n$ are exponentially distributed and $\lambda_i > 0$ are scalars, then $\min\{u_1/\lambda_1, \ldots, u_n/\lambda_n\}$ is distributed as $u/\lambda$, where $\lambda = \sum_i \lambda_i$. This property was used to estimate the $p$-norm of a vector, $p > 2$, in an elegant work of Andoni \cite{Andoni12}. In fact, by replacing the diagonal matrix $D$ in the sketch of \cite{MM12} with a diagonal matrix with entries $1/u_i^{1/p}$ for exponential random variables $u_i$, the sketch coincides with the sketch of Andoni, up to the setting of $t$. Importantly, this new setting of $D$ has no restriction on $p \in [1, \infty)$. We note that while Andoni's analysis for vector norms requires the variance of $1/u_i^{1/p}$ to exist, which requires $p > 2$, in our setting this restriction can be removed. If $X \sim 1/u^{1/p}$, then $X^p$ is only heavy-tailed {\it in one direction}, while the lower tail is exponentially decreasing. This results in a simpler analysis than \cite{MM12} for $1 \le p < 2$ and an improved distortion. The analysis of the expansion follows from the properties of a well-conditioned basis and is by now standard \cite{SW11,MM12,CDMMMW12}, while for the contraction by observing that $S$ is an $\ell_2$-subspace embedding, for any fixed $x$, $\|SDMx\|_1 \geq \|SDMx\|_2 \geq \frac{1}{2} \|DMx\|_2 \geq \frac{1}{2} \|DMx\|_{\infty} \sim \|Mx\|_1/(2u)$, where $u$ is an exponential random variable. Given the exponential tail of $u$, the bound for all $x$ follows from a standard net argument. While this already improves the distortion of \cite{MM12}, a more refined analysis gives a distortion of $\tilde{O}(d^{3/2})\log^{1/2} n$ provided $d > \log n$. 

%By using random variables $1/u_i^{1/p}$ on the diagonal, we have bypassed the two restrictions of $p$-stable random variables in the context of regression. First, there is no restriction on $p$, and therefore this can be applied to $\ell_p$-regression for any $p \in [1, \infty)$. Second, in the analysis for regression, we now bound the concentration of $X^p = 1/u$, where $u$ is an exponential random variable. While we still have $\Pr[X^p > t \cdot \textrm{median}(X^p)]$ is still $\Omega(1/t$), now $\Pr[|X|^p < 1/t \cdot \textrm{median}(X^p)]$ is only $\exp(-t)$, that is, the lower tail is now exponentially small. This property ultimately leads to our improved distortion bound for $\ell_p$ subspace embeddings into $\ell_p$ for $p \in [1,2)$, and in turn our improved running time for $\ell_p$-regression for $p \in [1,2)$. We note that the use of exponential random variables by Andoni requires $p \geq 2$ since the variance of $1/u^{1/p}$ does not exist unless $p > 2$ and he embeds vectors into $\ell_\infty$ with constant distortion, while our analysis for $p \in [1,2)$ instead embeds vectors into $\ell_2$ with a distortion that depends on the dimension of the underlying subspace. To prove this works, we also use the fact that $S$ provides a subspace embedding for $\ell_2$, as shown in \cite{CW12} and further optimized in \cite{MM12,NN12}. We then embed the low-dimensional $\ell_2$ space into $\ell_1$ using the Fast Johnson Lindenstrauss Transform. This illustrates the versatility of exponential random variables.

For $p > 2$, we need to embed our subspace into $\ell_{\infty}$. A feature is that it implies one can obtain a $\poly(d)$-approximation to $\ell_p$-regression by solving an $\ell_{\infty}$-regression problem, in $O(\nnz(M)) + \poly(d)$ time. As $\ell_{\infty}$-regression can be solved with linear programming, this may result in significant practical savings over convex program solvers for general $p$. This is also why we use the $\ell_{\infty}$-estimator for vector $p$-norms rather than the estimators of previous works \cite{IW05,AKO11,BGKS06,BO10} which were not norms, and therefore did not have efficient optimization procedures, such as finding a well-conditioned basis, in the sketch space. 
%For instance, many of these estimation procedures use the median operation which is non-convex and therefore solving a regression problem in the sketch space with respect to such an estimator need not be in polynomial time. 
%
%While our running time to obtain a $(1+\eps)$-approximation for $\ell_p$-regression for $p > 2$ is $\tilde{O}(\nnz(M) + n^{1-2/p} d^{4+p/2}) + \poly(d))$, which is comparable to the previous $\tilde{O}(\nnz(M) + \poly(d))$ time algorithm of \cite{CW12} for $d < n^{\gamma}$ for sufficiently small constant $\gamma$. 
Our embedding is into $n^{1-2/p} \poly(d)$ dimensions, whereas previous work was into $n/\poly(d)$ dimensions. This translates into near-optimal communication and computation protocols for distributed $\ell_p$-regression for every $p$. A parallel least squares regression solver LSRN was developed in \cite{MSM11}, and the extension to $1 \leq p < 2$ was a motivation of \cite{MM12}. Our result gives the analogous result for every $2 < p < \infty$, which is near-optimal in light of an $\Omega(n^{1-2/p})$ sketching lower bound for estimating the $p$-norm of a vector over the reals \cite{pw12}. %\cite{SS02,bjks04}. 

%% file: preliminary.tex
\section{Preliminaries}
\label{sec:preliminary}
In this paper we only consider the real RAM model of computation, and state our running times in terms of the number of arithmetic operations.

Given a matrix $M \in \mathbb{R}^{n \times d}$, let $M_1, \ldots, M_d$ be the columns of $M$, and $M^1, \ldots, M^n$ be the rows of $M$. Define $\ell_i = \norm{M^i}_p\ (i = 1, \ldots, n)$, where the $\ell_i^p$ are known as the {\em leverage scores} of $M$. Let $\text{range}(M) = \{y\ |\ y = Mx, x \in \mathbb{R}^d\}$. W.l.o.g., we constrain $\norm{x}_1 = 1, x \in \mathbb{R}^d$; by scaling our results will hold for all $x \in \mathbb{R}^d$. Define $\norm{M}_p$ to be the element-wise $\ell_p$ norm of $M$. 
That is, $\norm{M}_p = (\sum_{i \in [d]} \norm{M_i}_p^p)^{1/p} = (\sum_{j \in [n]} \norm{M^j}_p^p)^{1/p}$.

Let $[n] = \{1, \ldots, n\}$. Let $\omega$ denote the exponent of matrix multiplication.

\subsection{Well-Conditioning of A Matrix}
\label{sec:well-conditioning}
We introduce two definitions on the well-conditioning of matrices.
\begin{definition}[($\alpha, \beta, p$)-well-conditioning \cite{DDHKM09}]
\label{def:well-condition-1}
Given a matrix $M \in \mathbb{R}^{n \times d}$ and $p \in [1, \infty)$, let $q$ be the dual norm of $p$, that is, $1/p + 1/q = 1$. We say $M$ is {\em $(\alpha, \beta, p)$-well-conditioned} if 
(1) $\norm{x}_q \le \beta \norm{Mx}_p$ for any $x \in \mathbb{R}^d$, and (2) $\norm{M}_p \le \alpha$.
Define $\Delta'_p(M) = \alpha \beta$.
\end{definition}

It is well known that the Auerbach basis \cite{Auerbach1930} (denoted by $A$ throughout this paper) for a $d$-dimensional subspace $(\mathbb{R}^n, \norm{\cdot}_p)$ is $(d^{1/p}, 1, p)$-well-conditioned. Thus by definition we have
$\norm{x}_q \le \norm{Ax}_p$ for any $x \in \mathbb{R}^d$, and $\norm{A}_p \le d^{1/p}$. In addition, the Auerbach basis also has the property that $\norm{A_i}_p = 1$ for all $i \in [d]$.

\begin{definition}[$\ell_p$-conditioning \cite{CDMMMW12}]
\label{def:well-condition-2}
Given a matrix $M \in \mathbb{R}^{n \times d}$ and $p \in [1, \infty)$, define $\zeta_p^{\max}(M) = \max_{\norm{x}_2 \le 1} \norm{Mx}_p$ and $\zeta_p^{\min}(M) = \min_{\norm{x}_2 \ge 1} \norm{Mx}_p$. Define $\Delta_p(M) = \zeta_p^{\max}(M) / \zeta_p^{\min}(M)$ to be the {\em $\ell_p$-norm condition number} of $M$.
\end{definition}

The following lemma states the relationship between the two definitions.
\begin{lemma}[\cite{DDHKM09}]
\label{lem:relation-well-condition}
Given a matrix $M \in \mathbb{R}^{n \times d}$ and $p \in [1, \infty)$, we have
$$d^{-\abs{1/2-1/p}} \Delta_p(M) \le \Delta'_p(M) \le d^{\max\{1/2, 1/p\}} \Delta_p(M).$$
\end{lemma}

\subsection{Oblivious Subspace Embeddings}
An oblivious subspace embedding (OSE) for the Euclidean norm, given a parameter $d$, is a distribution $\mathcal{D}$ over $m \times n$ matrices such that for any $d$-dimensional subspace $\mathcal{S} \subset \mathbb{R}^n$, with probability $0.99$ over the choice of $\Pi \sim \mathcal{D}$, we have
$$1/2 \cdot \norm{x}_2 \le \norm{\Pi x}_2 \le 3/2 \cdot \norm{x}_2, \quad \forall x \in \mathcal{S}.$$ Note that OSE's only work for the $2$-norm, while in this paper we  get similar results for $\ell_p$-norms for all $p \in [1, \infty) \backslash \{2\}$. Two important parameters that we want to minimize in the construction of OSE's are: (1) The number of rows of $\Pi$, that is, $m$. This is the dimension of the embedding. (2) The number of non-zero entries in the columns of $\Pi$, denoted by $s$. This affects the running time of the embedding.

In \cite{NN12}, buiding upon \cite{CW12}, several OSE constructions are given. In particular, they show that there exist OSE's with $(m, s) = \left(O(d^2), 1\right)$ and $(m, s) = \left(O(d^{1+\gamma}), O(1)\right)$ for any constant $\gamma > 0$ and $(m, s) = (\tilde{O}(d), \log^{O(1)}d)$. 
%Of course OSE's with larger $(m, s)$ also exist. 

\subsection{Distributions}

\paragraph{$p$-stable Distribution.}
We say a distribution $\mathcal{D}_p$ is $p$-stable, if for any vector $\alpha = (\alpha_1, \ldots, \alpha_n) \in \mathbb{R}^n$ and $X_1, \ldots, X_n \stackrel{i.i.d.}{\sim} \mathcal{D}_p$, we have
$\textstyle \sum_{i \in [n]} \alpha_i X_i \simeq \norm{\alpha}_p X$,  where $X \sim \mathcal{D}_p$. It is well-known that $p$-stable distribution exists if and only if $p \in [1,2]$ (see. e.g., \cite{Indyk06}). For $p = 2$ it is the Gaussian distribution and for $p = 1$ it is the Cauchy  distribution. 
We say a random variable $X$ is $p$-stable if $X$ is chosen from a $p$-stable distribution.

\paragraph{Exponential Distribution.}
An exponential distribution has support $x \in [0, \infty)$,  probability density function (PDF) $f(x) = e^{-x}$ and cumulative distribution function (CDF) $F(x) = 1 - e^{-x}$. We say a random variable $X$ is exponential if $X$ is chosen from the exponential distribution. 
\begin{property}
\label{prop:exp}
The exponential distribution has the following properties.
\begin{enumerate}
\item ({\bf max stability}) 
If $u_1, \ldots, u_n$ are exponentially distributed, and $\alpha_i > 0\ (i = 1, \ldots, n)$ are real numbers, then 
$\textstyle \max\{\alpha_1 / u_1, \ldots, \alpha_n / u_n\} \simeq \left(\sum_{i \in [n]} \alpha_i \right) \left/ u \right.$,
where $u$ is exponential.

\item ({\bf lower tail bound})
For any $X$ that is exponential, there exist absolute constants $c_e, c'_e$  such that, \\
$\min\{0.5, c'_e t\} \le \Pr[X \le t] \le c_e t, \ \ \forall t \ge 0$.
\end{enumerate}
\end{property}
The second property holds since the median of the exponential distribution is the constant $\ln 2$ (that is, $\Pr[x \le \ln 2] = 50\%$), and the PDFs on $x = 0, x = \ln 2$ are $f(0) = 1, f(\ln 2) = 1/2$, differing by a factor of $2$. Here we use that the PDF is monotone decreasing.

\paragraph{Reciprocal of Exponential to the $p$-th Power.}
Let $E_i \sim 1/U_i^p$ where $U_i$ is an exponential. We call $E_i$ reciprocal of exponential to the $p$-th power. The PDF  of $E_i$ is given by $g(x) = p x^{-(p+1)} e^{-1/x^p}$.

\medskip

%Given two random variables $X, Y$, we write $X \sim Y$ if $X$ and $Y$ have the same distribution. Given two random variables $X, Y$ chosen from two probability distributions, we say $X \succeq Y$ if for $\forall t \in \mathbb{R}$ we have $\Pr[X \ge t] \ge \Pr[Y \ge t]$. 

The following lemma shows the relationship between the $p$-stable distribution and the exponential distribution. 

\begin{lemma}
\label{lem:stable-exp}
Let $y_1, \ldots, y_d \geq 0$ be scalars. Let $z \in \{1,2\}$.
Let $E_1, \ldots, E_d$ be $d$ independent reciprocal of exponential random variables to the $p$-th power ($p \in (0, 2)$), 
%For non-negative coefficients $y_1, \ldots, y_d$, 
and let $X = (\sum_{i \in [d]} (y_i E_i)^z)^{1/z}$. Let $S_1, \ldots, S_d$ be $d$ independent $p$-stable
random variables, and let $Y = (\sum_{i \in [d]} (y_i \abs{S_i})^z)^{1/z}$. 
There is a constant $\gamma > 0$ for which for any $t > 0$,
$$\Pr[X \geq t] \leq \Pr[Y \geq \gamma t].$$
\end{lemma}

\begin{proof}
By Nolan (\cite{Nolan13}, Theorem 1.12), there exist constants $c_N, c_p, c'_p > 0$ such that the PDF $f(x)$ of the $p$-stable ($p \in (0, 2)$) distribution satisfies $$c_p x^{-(p+1)} \le f(x) \le c'_p x^{-(p+1)},$$ for $\forall {x} > c_N$. Also, $p$-stable distribution is continuous, bounded and symmetric with respect to $y$-axis~\footnote{See, e.g., \url{http://en.wikipedia.org/wiki/Stable_distribution}}.

We first analyze the PDF $h$ of $y_i^z \abs{S_i}^z$. Letting $t = y_i^z \abs{S_i}^z$, the inverse function is $\abs{S_i} = t^{1/z}/y_i$. Taking the derivative,
we have $\frac{d\abs{S_i}}{dt} = \frac{1}{z y_i} t^{1/z-1}$. Let $f(t) = c(t) \cdot t^{-(p+1)}$ be the PDF of the absolute value of a $p$-stable random variable, where $2c_p \le c(t) \le 2c'_p$ for $t > c_N$. We have by the change of variable technique,
\begin{eqnarray}
h(t) & = & c\left(\frac{t^{1/z}}{y_i}\right) \cdot \left(\frac{t^{1/z}}{y_i}\right)^{-(p+1)} \cdot \frac{1}{z y_i} \cdot t^{1/z-1} \nonumber \\
\label{eq:stable-1}
& \ge &  2c_p \cdot \frac{y_i^p}{z t^{\frac{p}{z}+1}} \quad \quad \text{ if } t > (c_N y_i)^z.
\end{eqnarray}

We next analyze the PDF $k$ of $u_i^z E_i^z$. Letting $t = y_i^z E_i^z$, the inverse function is $E_i = t^{1/z}/y_i$. Taking
the derivative, $\frac{dE_i}{dt} = \frac{1}{zy_i} t^{1/z-1}$. Letting $g(t) = p t^{-(p+1)} e^{-1/t^p}$ be the PDF of the reciprocal of exponential to the $p$-th power,
we have by the change of variable technique,
\begin{eqnarray}
k(t) & = & e^{-\left(\frac{t^{1/z}}{y_i}\right)^{-p}} \cdot p \cdot \left(\frac{t^{1/z}}{y_i}\right)^{-(p+1)} \cdot \frac{1}{z y_i} \cdot t^{1/z-1} \nonumber \\
& = & e^{-\left(\frac{y_i}{t^{1/z}}\right)^p} \cdot p \cdot \frac{y_i^p}{z t^{\frac{p}{z}+1}} \nonumber \\
\label{eq:exp-1}
&\le& p \cdot \frac{y_i^p}{z t^{\frac{p}{z}+1}} \quad \quad (e^{-x} \le 1 \text{ for } x \ge 0)
\end{eqnarray}
By (\ref{eq:stable-1}) and (\ref{eq:exp-1}), when $t > (c_N y_i)^z$, $$k(t) \le p \cdot \frac{y_i^p}{z t^{\frac{p}{z}+1}} \le 2c_p \cdot \frac{1}{\kappa^{\frac{p}{z}+1}} \cdot \frac{y_i^p}{z t^{\frac{p}{z}+1}}  \le \frac{h(\kappa t)}{\kappa}$$ for a sufficiently small constant $\kappa$. When $t \le (c_N y_i)^z = O(1)$, we also have $k(t) \le \frac{h(\kappa t)}{\kappa}$ for a sufficiently small constant $\kappa$.

We thus have,
\begin{eqnarray*}
\Pr[X \geq t] & = & \Pr[X^z \geq t^z]\\
& = & \Pr\left[\sum_{i=1}^d y_i^z E_i^z \geq t^z\right]\\
& = & \int_{\sum_{i=1}^d t_i \geq t^z} k(t_1) \cdots k(t_d) dt_1 \cdots dt_d\\
& \leq & \int_{\sum_{i=1}^d t_i \geq t^z} \kappa^{-d} h(\kappa t_1) \cdots h(\kappa t_d) dt_1 \cdots dt_d\\
& \leq & \int_{\sum_{i=1}^d s_i \geq \kappa t^z} f(s_1) \cdots f(s_d) ds_1 \cdots ds_d\\
& = & \Pr[Y^z \geq \kappa t^z]\\
& = & \Pr[Y \geq \kappa^{1/z} t],
\end{eqnarray*}
where we made the change of variables $s_i = \kappa t_i$. Setting $\gamma = \kappa^{1/z}$ completes the proof. 
\end{proof}

\begin{lemma}
\label{lem:tail-exp}
Let $U_1, \ldots, U_d$ be $d$ independent exponentials. Let $X = \sum_{i \in [d]} 1/U_i$. There is a constant $\gamma > 0$ for which for any $t \ge 1$,
$$\Pr[X \ge t d / \gamma] \le (1+o(1)) \log(td)/t.$$
\end{lemma}

\begin{proof}
Let $C_1, \ldots, C_d$ be $d$ independent Cauchy ($1$-stable)
random variables, and let $Y = \sum_{i \in [d]} |C_i|$.  By Lemma 2.3 in \cite{CDMMMW12} we have for any $t \ge 1$,
$$\Pr[Y \ge t d] \le (1+o(1))\log(td)/t.$$
This lemma then follows from Lemma~\ref{lem:stable-exp} (setting $z=1$, $p=1$, and $y_1 = \ldots = y_d = 1$).
\end{proof}

We next use Lemma \ref{lem:stable-exp} (setting $z = 2$, $p = 1$) to show a bound on $\Pr[Y \geq t]$ for $Y = (\sum_{i \in [d]} y_i^2 C_i^2)^{1/2}$, where we have replaced $p$-stable random variable $S_i$ with Cauchy ($1$-stable) random variable $C_i$. Let $y = (y_1, \ldots, y_d)$.

\begin{lemma}\label{lem:cauchy-l2}
There is a constant $c > 0$ so that for any $r > 0$,
$$\Pr[Y \geq r \norm{y}_1] \leq \frac{c}{r}.$$
\end{lemma}
\begin{proof}
For $i \in [d]$, let $\sigma_i \in \{-1,+1\}$ be i.i.d.\ random variables with $\Pr[\sigma_i = -1] = \Pr[\sigma_i = 1] = 1/2$. 
Let $Z = \sum_{i \in [d]} \sigma_i y_i C_i$. We will obtain tail bounds for $Z$ in two different ways, and use this to establish
the lemma.

On the one hand, by the $1$-stability of the Cauchy distribution, we have that $Z \sim \norm{y}_1 C$, where $C$ is a standard
Cauchy random variable. Note that this holds for any fixing of the $\sigma_i$. The cumulative distribution function
of the absolute value of Cauchy distribution is $F(z) = \frac{2}{\pi} \arctan(z).$ Hence for any $r > 0$,
\begin{eqnarray*}
\Pr[Z \geq r \norm{y}_1] & \le & \Pr[\abs{C} \geq r] = 1 - \frac{2}{\pi} \arctan(r).
\end{eqnarray*}
We can use the identity 
$$\arctan(r) + \arctan \left (\frac{1}{r} \right) = \frac{\pi}{2},$$
and therefore using the Taylor series for $\arctan$ for $r > 1$,
$$\arctan(r) \geq \frac{\pi}{2} - \frac{1}{r}.$$
Hence,
\begin{eqnarray}\label{eqn:upper}
\Pr[Z \geq r \norm{y}_1] \leq \frac{2}{\pi r}.
\end{eqnarray}
On the other hand, for any fixing of $C_1, \ldots, C_d$, we have
$${\bf E}[Z^2] = \sum_{i \in [d]} y_i^2 C_i^2,$$
and also 
$${\bf E}[Z^4] = 3 \sum_{i \neq j \in [d]} y_i^2 y_j^2 C_i^2 C_j^2 + \sum_{i \in [d]} y_i^4 C_j^4.$$
We recall the Paley-Zygmund inequality.
\begin{fact}
If $R \geq 0$ is a random variable with finite variance, and $0 < \theta < 1$, then
$$\Pr[R \geq \theta {\bf E}[R]] \geq (1-\theta)^2 \cdot \frac{{\bf E}[R]^2}{{\bf E}[R^2]}.$$
\end{fact}
Applying this inequality with $R = Z^2$ and $\theta = 1/2$, we have
\begin{eqnarray*}
\Pr\left[Z^2 \geq \frac{1}{2} \cdot \sum_{i \in [d]} y_i^2 C_i^2\right] & \geq & 
\frac{1}{4} \cdot \frac{\left (\sum_{i \in [d]} y_i^2 C_i^2 \right )^2}{3 \sum_{i \neq j \in [d]} y_i^2 y_j^2 C_i^2 C_j^2 + \sum_{i \in [d]} y_i^4 C_i^4} \geq \frac{1}{12},
\end{eqnarray*}
or equivalently 
\begin{eqnarray}\label{eqn:applyPZ}
\Pr\left[Z \geq \frac{1}{\sqrt{2}} \left(\sum_{i \in [d]} y_i^2 C_i^2\right)^{1/2}\right] \geq \frac{1}{12}.
\end{eqnarray}
Suppose, towards a contradiction, that $\Pr[Y \geq r \norm{y}_1] \geq c/r$ for a sufficiently large constant $c > 0$. By independence of the $\sigma_i$ and the $C_i$, by
(\ref{eqn:applyPZ}) this implies
$$\Pr\left[Z \geq \frac{r \norm{y}_1}{\sqrt{2}}\right] \geq \frac{c}{12r}.$$
By (\ref{eqn:upper}), this is a contradiction for $c > \frac{24}{\pi}$. It follows that $\Pr[Y \geq r \norm{y}_1] < c/r$, as desired. 
\end{proof}

\begin{corollary}\label{cor:upper}
Let $y_1, \ldots, y_d \geq 0$ be scalars. 
Let $U_1, \ldots, U_d$ be $d$ independendent exponential random variables, 
%For non-negative coefficients $y_1, \ldots, y_d$, 
and let $X = (\sum_{i \in [d]} y_i^2/U_i^2)^{1/2}$. There is a constant $c > 0$ for which for any $r > 0$,
$$\Pr[X > r \norm{y}_1] \leq c/r.$$
\end{corollary}
\begin{proof}
The corollary follows by combining Lemma \ref{lem:stable-exp} with Lemma \ref{lem:cauchy-l2}, and rescaling the constant $c$ from Lemma \ref{lem:cauchy-l2}
by $1/\gamma$, where $\gamma$ is the constant of Lemma \ref{lem:stable-exp}. 
\end{proof}

\paragraph{Conventions.}
In the paper we will define several events $\mathcal{E}_0, \mathcal{E}_1, \ldots$ in the early analysis, which we will condition on in the later analysis. Each of these events holds with probability $0.99$, and there will be no more than ten of them. Thus by a union bound all of them hold simultaneously with probability $0.9$. Therefore these conditions will not affect our overall error probability by more than $0.1$. 

\paragraph{Global Parameters.} We set a few parameters which will be used throughout the paper:
%We briefly comment below where these parameters will be used. We believe this list will be helpful for readers when going through the analysis of the paper. Let $c_1$ be a sufficiently large constant and $c_2 = 2c_1$.
%\begin{itemize}
%\item $\gamma = 100 c_e d$. Used when analyzing the left tail of the exponential distribution. 
$\rho = c_1 d \log d$; %Used when analyzing the upper tail of the exponential distribution. 
$\iota = 1/(2\rho^{1/p})$; %Used when bounding the noise in buckets of hash tables.
$\eta = c_2  d \log d  \log n$; %Used when analyzing the probability that the noise in buckets of hash tables is small.
$\tau = \iota / (d\eta)$. %A threshold for ratios between leverage scores and exponentials.
%\end{itemize}
%\qinsays{Good to write these here?}

%%\qinsays{Talk about the tail bounds for exponential, Cauchy, Gaussian and $p$-stable.}

%Let $U$ be exponential, $X_p$ be $p$-stable and $C$ be Cauchy. 
%$$\abs{C} \succeq \kappa_1 \abs{X_p}^p \succeq \kappa_2 \abs{G}^2 \succeq \kappa_3 \abs{U}$$ for constants $\kappa_1, \kappa_2, \kappa_3$.

%% file: big-p.tex
\section{$p$-norm with $p > 2$}
\label{sec:big}
\subsection{Algorithm}
We set the subspace embedding matrix $\Pi = S D$, where $D \in \mathbb{R}^{n \times n}$ is a diagonal matrix with $1/u_1^{1/p}, \ldots, 1/u_n^{1/p}$ on the diagonal such that all $u_i\ (i = 1, 2, \ldots, n)$ are i.i.d.\ exponentials. And $S$ is an $(m,s)$-OSE with $(m, s) = \left(6 n^{1-2/p} \eta / \iota^2 + d^{5+4p}, 1\right)$. More precisely, we pick random hash functions $h : [n] \to [m]$ and $\sigma: [n] \to \{-1, 1\}$. For each $i \in [n]$, we set $S_{h(i), i} = \sigma(i)$. Since $m = \omega(d^2)$, by \cite{NN12} such an $S$ is an OSE. 

\subsection{Analysis}
In this section we prove the following Theorem.
\begin{theorem}
\label{thm:big-p}
Let $A \in \mathbb{R}^{n \times d}$ be an Auerbach basis of a $d$-dimensional subspace of $(\mathbb{R}^n, \norm{\cdot}_p)$. Given the above choices of $\Pi \in \mathbb{R}^{(6 n^{1-2/p} \eta / \iota^2 + d^{5+4p}) \times n}$, for any $p > 2$ we have
\begin{equation*}
\Omega(1 / (d \log d)^{1/p}) \cdot \norm{Ax}_p \le \norm{\Pi A x}_\infty \le O((d \log d)^{1/p}) \cdot \norm{Ax}_p, \quad \forall x \in \mathbb{R}^d. 
\end{equation*}
\end{theorem}
\begin{remark}
\label{rem:auerbach}
Note that since the inequality holds for all $x \in \mathbb{R}^d$, this theorem also holds if we replace the Auerbach basis $A$ by any matrix $M$ whose column space is a $d$-dimensional subspace of $(\mathbb{R}^n, \norm{\cdot}_p)$.
\end{remark}

\begin{property}
\label{prop:y}
Let $A \in \mathbb{R}^{n \times d}$ be a $(d^{1/p}, 1, p)$-well-conditioned Auerbach basis. For an $x \in \mathbb{R}^d$, let $y = Ax \in \text{range}(A) \subseteq \mathbb{R}^n$.  Each such $y$ has the following properties. Recall that we can assume $\norm{x}_1 = 1$.
\begin{enumerate}
\item $\norm{y}_p \le \sum_{i \in d} \norm{A_i}_p \cdot \abs{x_i} = \norm{x}_1 = 1$.

\item $\norm{y}_p = \norm{Ax}_p \ge \norm{x}_q \ge \norm{x}_1 / d^{1-1/q} = 1/d^{1/p}$. 

\item For all $i \in [n]$,
$\abs{y_i} = \abs{(A^i)^T x} \le \norm{A^i}_1 \cdot \norm{x}_\infty \le d^{1-1/p} \norm{A^i}_p \cdot \norm{x}_1 = d^{1-1/p} \ell_i.$
\end{enumerate}
\end{property}

Let $H$ be the set of indices $i \in [n]$ such that $\ell_i / u_i^{1/p} \ge \tau$. Let $L = [n] \backslash H$. Then
\begin{eqnarray*} 
\E[\abs{H}]  &=& \textstyle \sum_{i \in [n]} \Pr[\ell_i / u_i^{1/p} \ge \tau] \\
&=& \textstyle \sum_{i \in [n]} \Pr[u_i \le \ell_i^p / \tau^p] \\
&\le& \textstyle \sum_{i \in [n]}  c_e \ell_i^p / \tau^p \quad \text{(Property~\ref{prop:exp})}\\ 
&\le& \textstyle c_e d / \tau^p. \quad (\sum_{i \in [n]} \ell_i^p = \norm{A}_p^p \le d)
\end{eqnarray*}
Therefore with probability $0.99$, we have $\abs{H} \le 100 c_e d / \tau^p$. Let $\mathcal{E}_0$ denote this event, which we will condition on in the rest of the proof.

For a $y \in \text{range}(A)$, let $w_i = 1/u_i^{1/p} \cdot y_i$. For all $i \in L$, we have 
\begin{equation*}
\abs{w_i} = 1/u_i^{1/p} \cdot \abs{y_i} \le d^{1-1/p} \ell_i / u_i^{1/p} < d^{1-1/p} \tau \le d^{1-1/p} \tau \cdot d^{1/p} \norm{y}_p = d \tau \norm{y}_p. 
\end{equation*}
In the first and third inequalities we use Property~\ref{prop:y}, and the second inequality follows from the definition of $L$.
For $j \in [m]$, let 
$$z_j(y) = \sum_{i : (i \in L) \wedge (h(i)=j)} \sigma(j) \cdot w_i.$$ 
Define $\mathcal{E}_1$ to be the event that for all $i, j \in H$, we have $h(i) \neq h(j)$. The rest of the proof conditions on $\mathcal{E}_1$. The following lemma is implicit in \cite{Andoni12}. 
%See Section~\ref{sec:app-2-1} for a sketch of the proof.

\begin{lemma}[\cite{Andoni12}]
\label{lem:alex}

\begin{enumerate}
\item Assuming that $\mathcal{E}_0$ holds, $\mathcal{E}_1$ holds with probability at least $0.99$. 

\item For any $\iota > 0$, for all $j \in [m]$,
\begin{equation*}
\Pr[\abs{z_j(y)} \ge \iota \norm{y}_p] \le \exp \left[- \frac{\iota^2/2}{n^{1-2/p}/m + \iota d\tau/3} \right] = e^{-\eta}.
\end{equation*}
\end{enumerate}
\end{lemma}
\begin{proof} (sketch, and we refer readers to \cite{Andoni12} for the full proof). The first item simply follows from the birthday paradox; note that by our choice of $m$ we have $\sqrt{m} = \omega(d / \tau^p)$. For the second item, we use Bernstein's inequality to show that for each $j \in [m]$, $z_j(y)$ is tightly concentrated around its mean, which is $0$.
\end{proof}

\subsubsection{No Overestimation}
%First, recall that $\norm{DA_i}_\infty \sim \norm{A_i}_p / v_i^{1/p}$ for an exponential $v_i$. Let $\mathcal{E}_2$ be the event that $v_i \ge 1/\gamma$ for all $i \in [d]$. Since $\Pr[v_i < 1/\gamma] \le c_e/\gamma$, $\mathcal{E}_2$ holds with probability$$1 - \sum_{i \in [d]} \Pr[v_i < 1/\gamma] \ge 1 - c_e d/\gamma \ge 0.99.$$ 
By Lemma~\ref{lem:alex} we have that with probability $(1 - m \cdot d \cdot e^{-\eta}) \ge 0.99$, $\max_{j \in [m]} z_j(A_i) \le \iota \norm{A_i}_p = \iota$ for all $i \in [d]$. Let $\mathcal{E}_2$ denote this event, which we condition on. Note that $A_i \in \text{range}(A)$ for all $i \in [d]$. Thus,
\begin{eqnarray}
\norm{SDAx}_\infty &\le& \textstyle \sum_{i \in [d]}\norm{SDA_i}_\infty \cdot \abs{x_i} \nonumber \\
&\le& \textstyle \sum_{i \in [d]} \left(\norm{DA_i}_\infty + \max_{j \in [m]} z_j(A_i)\right) \cdot \abs{x_i} \label{eq:b-1} \quad (\text{conditioned on } \mathcal{E}_1) \nonumber \\
%&=& \textstyle \sum_{i \in [d]} \left(\norm{DA_i}_\infty \cdot \abs{x_i}\right) + \sum_{i \in [d]} \left(\max_{j \in [m]} z_j(A_i) \cdot \abs{x_i} \right) \nonumber  \\
&\le& \textstyle \sum_{i \in [d]} (\norm{DA_i}_\infty \cdot \abs{x_i})+ \iota \cdot \norm{x}_1, \quad (\text{conditioned on } \mathcal{E}_2) \label{eq:b-10}
\end{eqnarray}
Let $v_i = \norm{DA_i}_\infty$ and $v = \{v_1, \ldots, v_d\}$. By H\"older's inequality, we have $$\textstyle \sum_{i \in [d]} (\norm{DA_i}_\infty \cdot \abs{x_i}) = \sum_{i \in [d]} (v_i \cdot \abs{x_i}) \le \norm{v}_p \norm{x}_q.$$ 
We next bound $\norm{v}_p$:
$$\textstyle \norm{v}_p^p = \sum_{i \in [d]} \norm{DA_i}_\infty^p \sim \sum_{i \in [d]} \norm{A_i}_p^p / u_i = \sum_{i \in [d]} 1 / u_i,$$
where each $u_i\ (i \in [d])$ is an exponential. By Lemma~\ref{lem:tail-exp} we know that with probability $0.99$, 
$\sum_{i \in [d]} 1 / u_i \le 200/\kappa_1 \cdot d\log d$, thus $\norm{v}_p \le (200 / \kappa_1 \cdot d\log d)^{1/p}$. Denote this event by $\mathcal{E}_3$ which we condition on. Thus,
\begin{eqnarray}
(\ref{eq:b-10}) & \le & \norm{v}_p \norm{x}_q  + \iota  \norm{x}_1 \nonumber \\
&\le& (200/\kappa_1 \cdot d\log d)^{1/p} \norm{x}_q + \iota d^{1-1/q} \norm{x}_q \quad \text{(conditioned on $\mathcal{E}_3)$} \nonumber \\
&\le &  2(200/\kappa_1 \cdot d\log d)^{1/p} \norm{x}_q \quad (\iota < 1/d^{1/p}) \nonumber \\
& \le & 2(200/\kappa_1 \cdot d\log d)^{1/p} \cdot \norm{Ax}_p. \quad (A \text{ is $(d^{1/p}, 1, p)$-well-conditioned}) \label{eq:e-1}
\end{eqnarray}

\subsubsection{No Underestimation}    
\label{sec:big-p-no-under}
In this section we lower bound $\norm{SDAx}_\infty$, or $\norm{SDy}_\infty$, for all $y \in \text{range}(A)$. 
For a fixed $y \in \text{range}(A)$, by the triangle inequality
\begin{eqnarray*}
\norm{SDy}_\infty &\ge& \textstyle \norm{Dy}_\infty - \max_{j \in [m]} z_j(y).
\end{eqnarray*}
By Lemma~\ref{lem:alex} we have that with probability $(1 - m \cdot e^{-\eta})$, $z_j(y) \le \iota \norm{y}_p$ for all $j \in [m]$.  We next bound $\norm{Dy}_\infty$. By Property~\ref{prop:exp}, it holds that $\norm{Dy}_\infty \sim \norm{y}_p / v^{1/p}$, where $v$ is an exponential. Since $\Pr[v \ge \rho] \le e^{-\rho}$ for an exponential $v$, with probability $(1 - e^{-\rho})$ we have
\begin{eqnarray}
\label{eq:a-0}
\norm{Dy}_\infty & \ge & 1/\rho^{1/p} \cdot \norm{y}_p,  \quad \forall y \in \text{range}(A). 
\end{eqnarray}
Therefore, with probability $(1 - m \cdot e^{-\eta} - e^{-\rho}) \ge (1 - 2 e^{-\rho})$, 
\begin{eqnarray}
\label{eq:a-1}
\norm{SDy}_\infty &\ge& \norm{Dy}_\infty - \iota \norm{y}_p \ge 1/(2\rho^{1/p}) \cdot \norm{y}_p. 
\end{eqnarray}
%Let $\mathcal{E}_4$ be this event which we condition on in the rest of the proof. 

Given the above ``for each" result (for each $y$, the bound holds with probability $1 - 2e^{-\rho}$), we next use a standard net-argument to show 
\begin{equation}
\label{eq:a-3}
\norm{SDy}_\infty \ge \Omega\left(1 / \rho^{1/p} \cdot \norm{y}_p\right),  \quad \forall y \in \text{range}(A).
\end{equation}
%Due to space constraints, we leave the arguments to Appendix~\ref{sec:app-2-2}.

Let the ball $B = \{y \in \mathbb{R}^n\ |\ y = Ax, \norm{x}_1 = 1\}$. By Property~\ref{prop:y} we have $\norm{y}_p \le 1$ for all $y \in B$. Call $B_\eps \subseteq B$ an $\eps$-net of $B$ if for any $y \in B$, we can find a $y' \in B_\eps$ such that $\norm{y - y'}_p \le \eps$. It is well-known that $B$ has an $\eps$-net of size at most $(3/\eps)^d$ \cite{BLM89}.  We choose $\eps = 1/(8(200/\kappa_1 \cdot \rho d^2 \log d)^{1/p}$, then with probability 
\begin{eqnarray}
1 - 2e^{-\rho} \cdot (3/\eps)^d & = & 1 - 2 e^{-c_1 d \log d} \cdot \left(24 (200/\kappa_1 \cdot c_1 d \log d \cdot d^2  \log d)^{1/p}\right)^d \nonumber \\
&\ge& 0.99, \quad (c_1 \text{ sufficiently large}) \nonumber
\end{eqnarray}
$\norm{SDy'}_\infty \ge 1/(2\rho^{1/p}) \cdot \norm{y'}_p $ holds for all $y' \in B_\eps$. Let $\mathcal{E}_4$ denote this event which we condition on.

Now we consider $\{y\ |\ y \in B \backslash B_\eps\}$. Given any $y \in B \backslash B_\eps$, let $y' \in B_\eps$ such that $\norm{y - y'}_p \le \eps$. By the triangle inequality we have
\begin{eqnarray}
\label{eq:a-2}
\norm{SDy}_\infty &\ge& \norm{SDy'}_\infty - \norm{SD(y  -y')}_\infty.
\end{eqnarray}
Let $x'$ be such that $Ax' = y'$. Let $\tilde{x} = x - x'$. Let $\tilde{y} = A\tilde{x} = y - y'$. Thus $\norm{\tilde{y}}_p  = \norm{A\tilde{x}}_p \le \eps$.
\begin{eqnarray}
\norm{SD(y - y')}_\infty & = & \norm{SDA\tilde{x}}_\infty \nonumber \\
&\le& 2(200/\kappa_1 \cdot d\log d)^{1/p} \cdot \norm{A\tilde{x}}_p \quad (\text{by } (\ref{eq:e-1})) \nonumber \\
&\le& 2(200/\kappa_1 \cdot d\log d)^{1/p}  \cdot \eps. \nonumber \\
&\le& 2(200/\kappa_1 \cdot d\log d)^{1/p} \cdot \eps \cdot d^{1/p} \cdot \norm{y}_p \quad (\text{by Property~\ref{prop:y}}) \nonumber \\
&=& 1/(4\rho^{1/p}) \cdot \norm{y}_p \quad (\eps = 1/(8(200/\kappa_1 \cdot \rho d^2 \log d)^{1/p})  \label{eq:g-1}
\end{eqnarray}
By (\ref{eq:a-1}),  (\ref{eq:a-2}) ,  (\ref{eq:g-1}), conditioned on $\mathcal{E}_4$,  we have for all $y \in \text{range}(A)$, it holds that
\begin{eqnarray*}
\norm{SDy}_\infty \ge 1/(2\rho^{1/p}) \cdot \norm{y}_p - 1/(4\rho^{1/p}) \cdot \norm{y}_p \ge 1/(4\rho^{1/p}) \cdot \norm{y}_p.
\end{eqnarray*}

Finally, Theorem \ref{thm:big-p} follows from inequalities (\ref{eq:e-1}), (\ref{eq:a-3}), and our choice of $\rho$.

%% file: small-p.tex
\section{$p$-norm with $1 \le p \le 2$}
\label{sec:small}
\subsection{Algorithm}
Our construction of the subspace embedding matrix $\Pi$ is  similar to that for  $p$-norms with $p > 2$: We again set $\Pi = S D$, where $D$ is an $n \times n$ diagonal matrix with $1/u_1^{1/p}, \ldots, 1/u_n^{1/p}$ on the diagonal, where $u_i\ (i = 1, \ldots, n)$ are i.i.d.\ exponentials. The difference is that this time we choose $S$ to be an $(m,s)$-OSE with $(m, s) = \left(O(d^{1+\gamma}), O(1)\right)$ from \cite{NN12} ($\gamma$ is an arbitrary small constant). More precisely, we first pick random hash functions $h : [n] \times [s] \to [m/s], \sigma: [n] \times [s] \to \{-1, 1\}$. 
%Both are $O(\log n)$-wise independent. 
For each $(i, j) \in [n] \times [s]$, we set $S_{(j-1)s+h(i,j), i} = \sigma(i,j)/\sqrt{s}$, where $\sqrt{s}$ is just a normalization factor.

\subsection{Analysis}
In this section we prove the following theorem.
\begin{theorem}
\label{thm:small-p}
Let $A$ be an Auerbach basis of a $d$-dimensional subspace of $(\mathbb{R}^n, \norm{\cdot}_p)\ (1 \le p < 2)$. Given the above choices of $\Pi \in \mathbb{R}^{O(d^{1+\gamma}) \times n}$, with probability $2/3$,
\begin{equation*}
\Omega\left(\max\left\{1 / {(d \log d \log n)^{\frac{1}{p}-\frac{1}{2}}}, 1 / (d \log d)^{1/p}\right\}\right) \cdot \norm{Ax}_p \le \norm{\Pi A x}_2 \le O((d \log d)^{1/p}) \cdot \norm{Ax}_p, \ \  \forall x \in \mathbb{R}^d. 
\end{equation*}
\end{theorem}

Same as Remark~\ref{rem:auerbach}, since the inequality holds for all $x \in \mathbb{R}^d$, the theorem holds if we replace the Auerbach basis $A$ by any matrix $M$ whose column space is a $d$-dimensional subspace of $(\mathbb{R}^n, \norm{\cdot}_p)$. The embedding $\Pi M$ can be computed in time $O(\nnz(M) + \tilde{O}(d^{2+\gamma}))$.

\begin{remark}
\label{rem:inter-norm}
Using the inter-norm inequality $\norm{\Pi A x}_2 \le \norm{\Pi A x}_p \le d^{(1+\gamma)(1/p-1/2)}\norm{\Pi A x}_2, \  \forall p \in [1, 2)$, we can replace the $2$-norm estimator by the $p$-norm estimator in Theorem \ref{thm:small-p} by introducing another $d^{(1+\gamma)(1/p-1/2)}$ factor in the distortion. We will remove this extra factor for $p = 1$ below. 
\end{remark}

In the rest of the section we prove Theorem~\ref{thm:small-p}.
Define $\mathcal{E}_5$ to be the event that $\norm{S D A x}_2 = (1 \pm 1/2) \norm{D A x}_2$ for any $x \in \mathbb{R}^d$. Since $S$ is an OSE, $\mathcal{E}_5$ holds with probability $0.99$. 

\subsubsection{No Overestimation}
We can write $S = \frac{1}{\sqrt{s}}  (S_1, \ldots, S_s)^T$, where each $S_i \in \mathbb{R}^{(m/s) \times n}$ with one $\pm 1$ on each column in a random row. Let $S' \sim S$, and we also write $S' = \frac{1}{\sqrt{s}}  (S'_1, \ldots, S'_s)^T$.
 Let $D' \in \mathbb{R}^{n \times n}$ be a diagonal matrix with i.i.d.\ $p$-stable random variables on the diagonal. Let $\mathcal{E}'_5$ to be the event that $\norm{S' D' A x}_2 = (1 \pm 1/2) \norm{D' A x}_2$ for any $x \in \mathbb{R}^d$, which holds with probability $0.99$. 
 
For any $x \in \mathbb{R}^d$, let $y = Ax \in \mathbb{R}^n$. Let $\mathcal{E}_6$ be the event that for all $i \in [s]$, $\norm{S'_i D'y}_p \le c_4 (d\log d)^{1/p} \cdot \norm{y}_p$ for all $y \in \text{range}(A)$, where $c_4$ is a constant. Since $s = O(1)$ and $S'_1, \ldots, S'_s$ are independent, we know by \cite{MM12} (Sec. A.2 in \cite{MM12}) that $\mathcal{E}_6$ holds with probability $0.99$. 

The following deductions link the tail of $\norm{SDy}_2$ to the tail of $\norm{S'D'y}_2$.
\begin{eqnarray}
\Pr_{S,D}[\norm{SDy}_2 > t] &=& \Pr_D [\Pr_S [\norm{SDy}_2 > t] ] \nonumber\\
&\le& \Pr_D [\Pr_S [\norm{SDy}_2 > t\ |\ \mathcal{E}_5] \cdot \Pr_S[\mathcal{E}_5] + \Pr_S[\neg \mathcal{E}_5]] \nonumber\\
&\le& \Pr_D [\norm{Dy}_2 > t/2] \cdot 0.99 + 0.01 \nonumber\\
&\le& \Pr_{D'} \left[\norm{D'y}_2 > \gamma t/2\right] + 0.01 \quad (\text{Lemma~\ref{lem:stable-exp}}) \nonumber\\
&\le& \left(\Pr_{D'} \left[\Pr_{S'} \left[\norm{S' D' y}_2 > \gamma t/4\ |\ \mathcal{E}'_5 \right]\right] \cdot \Pr_{S'} [\mathcal{E}'_5] + \Pr_{S'} [\neg \mathcal{E}'_5] \right) + 0.01 \nonumber\\
&\le& \left(\Pr_{D'} \left[\Pr_{S'} \left[\norm{S' D' y}_2 > \gamma t/4\ |\ \mathcal{E}'_5 \right]\right] \cdot 0.99 + 0.01 \right) + 0.01 \nonumber\\
&\le& \Pr_{D'} \left[\Pr_{S'} \left[\norm{S' D' y}_2 > \gamma t/4\ |\ \mathcal{E}'_5 \right]\right] +0.02 \nonumber\\
&\le& \Pr_{D', S'} \left[\norm{S' D' y}_2 > \gamma t/4\ |\ \mathcal{E}'_5, \mathcal{E}_6 \right] +0.03. \label{eq:stable-exp-1}
\end{eqnarray}
We next analyze $\norm{S' D' y}_2$ conditioned on $\mathcal{E}_6$.
\begin{eqnarray}
%\norm{SDy}_2 &\le& 3/2 \cdot \norm{Dy}_2 \quad (\text{conditioned on } \mathcal{E}_5) \nonumber \\ 
%&\preceq& 3/2 \cdot \kappa_p \norm{D'y}_2 \quad (\text{Lemma~\ref{lem:stable-exp}, set $z = 2$}) \nonumber \\
\norm{S' D' y}_2  &\le& \norm{S' D' y}_p \nonumber \\
&\le& \frac{1}{\sqrt{s}} \sum_{i \in [s]} \norm{S'_i D'y}_p  \quad (\text{triangle inequality}) \nonumber \\
&\le& \frac{1}{\sqrt{s}} \cdot s \cdot c_4 (d\log d)^{1/p} \cdot \norm{y}_p \quad (\text{conditioned on } \mathcal{E}_6) \nonumber \\
&\le& c'_5 (d\log d)^{1/p} \cdot \norm{y}_p, \quad (c'_5 \text{ sufficiently large; note that } s = O(1)) \label{eq:exp-tail-1} 
\end{eqnarray}
Setting $\gamma t/4 = c'_5 (d \log d)^{1/p} \cdot \norm{y}_p$, or, $t = c_5 (d \log d)^{1/p} \cdot \norm{y}_p$ where $c_5 = 4c'_5/\gamma$, we have
\begin{eqnarray*}
&&\Pr_{S,D}[\norm{SDy}_2 > c_5 (d \log d)^{1/p} \cdot \norm{y}_p] \nonumber \\
&\le& \Pr_{D', S'} \left[\norm{S' D' y}_2 >  c'_5 (d \log d)^{1/p} \cdot \norm{y}_p\ |\ \mathcal{E}'_5, \mathcal{E}_6 \right] +0.03 \quad \text{(by (\ref{eq:stable-exp-1}))} \nonumber \\
&=& 0.03. \quad \text{(by (\ref{eq:exp-tail-1}))}  
\end{eqnarray*}
Let $\mathcal{E}_8$ be the event that 
\begin{equation}
\norm{S D y}_2 \le  c_5 (d \log d)^{1/p} \cdot \norm{y}_p,
\label{eq:h-3}
\end{equation}
which we condition on in the rest of the analysis. Note that $\mathcal{E}_8$ holds with probability $0.97$ conditioned on $\mathcal{E}'_5$ and $\mathcal{E}_6$ holds.

\subsubsection{No Underestimation}
\label{sec:small-p-no-under}
For any $x \in \mathbb{R}^d$, let $y = Ax \in \mathbb{R}^n$.
\begin{eqnarray}
\norm{SDy}_2 &\ge& 1/2 \cdot \norm{Dy}_2 \quad (\text{conditioned on } \mathcal{E}_5) \nonumber \\ 
&\ge& 1/2 \cdot \norm{Dy}_\infty \sim 1/2 \cdot \norm{y}_p / u \quad \text{($u$ is exponential)} \nonumber \\
&\ge& 1/2 \cdot 1/\rho^{1/p} \cdot \norm{y}_p.  \quad (\text{By (\ref{eq:a-0}), holds w.pr. $(1 - e^{-\rho})$}) \label{eq:f-3}
\end{eqnarray}

Given this ``for each" result, we again use a net-argument to show 
\begin{equation}
\label{eq:k-3}
\norm{SDy}_2 \ge \Omega\left(1 / \rho^{1/p} \cdot \norm{y}_p\right) = \Omega\left(1 / (d \log d)^{1/p}\right)  \cdot \norm{y}_p,  \quad \forall y \in \text{range}(A).
\end{equation}
%Due to space constraints, we leave it to Appendix~\ref{sec:app-3-1}.

Let the ball $B = \{y \in \mathbb{R}^n\ |\ y = Ax, \norm{y}_p \le 1\}$. Let $B_\eps \subseteq B$ be an $\eps$-net of $B$ with size at most $(3/\eps)^d$. We choose $\eps =  1/(4c_5(\rho d^2 \log d)^{1/p})$. Then with probability $1 - e^{-\rho} \cdot (3/\eps)^d \ge 0.99$, $\norm{SDy'}_2 \ge 1/(2\rho^{1/p}) \cdot \norm{y'}_p$ holds for all $y' \in B_\eps$. Let $\mathcal{E}_7$ denote this event which we condition on. For $y \in B \backslash B_\eps$, let $y' \in B_\eps$ such that $\norm{y - y'}_p \le \eps$. By the triangle inequality,
\begin{equation}
\label{eq:k-1}
\norm{SDy}_2  \ge \norm{SDy'}_2 - \norm{SD(y  -y')}_2.
\end{equation}
By (\ref{eq:h-3}) we have 
\begin{eqnarray}
\norm{SD(y - y')}_2 &\le& c_5 (d\log d)^{1/p} \cdot \norm{y - y'}_p \nonumber \\ 
&\le& c_5 (d\log d)^{1/p} \cdot \eps \nonumber \\ 
&\le& c_5 (d\log d)^{1/p} \cdot \eps \cdot d^{1/p} \norm{y}_p \nonumber \\ &=& 1/(4\rho^{1/p}) \cdot \norm{y}_p. \label{eq:k-2}
\end{eqnarray}
By (\ref{eq:f-3}) (\ref{eq:k-1}) and (\ref{eq:k-2}), conditioned on $\mathcal{E}_7$, we have for all $y \in \text{range}(A)$, it holds that
\begin{equation*}
\norm{SDy}_2 \ge 1/(2\rho^{1/p}) \cdot \norm{y}_p - 1/(4\rho^{1/p}) \cdot \norm{y}_p  \ge 1/(4\rho^{1/p}) \cdot \norm{y}_p.
\end{equation*}

In the case when $d \ge \log^{2/p-1} n$, using a finer analysis we can show that 
\begin{eqnarray}
\label{eq:finer}
\norm{SDy}_2 \ge \Omega\left(1 \left/ {(d \log d \log n)^{\frac{1}{p}-\frac{1}{2}}} \right. \right) \cdot \norm{y}_p, \quad \forall y \in \text{range}(A).
\end{eqnarray}
%Due to the space constraints, we leave the improved analysis to Section~\ref{sec:improve-small-p}.
The analysis will be given in the Section~\ref{sec:improve-small-p}.

Finally, Theorem \ref{thm:small-p} follows from (\ref{eq:h-3}), (\ref{eq:k-3}), (\ref{eq:finer}) and our choices of $\rho$.

\input{improve-L1}

\subsection{An Improved Dilation for $\ell_1$ Subspace Embeddings}
We can further improve the dilation for $\ell_1$ using the $1$-norm estimator in Remark~\ref{rem:inter-norm}. Let $S \in \mathbb{R}^{\tilde{O}(d) \times n}$ be a $(\tilde{O}(d), \log^{O(1)} d)$-OSE, which can be written as $\frac{1}{\sqrt{s}}(S_1, \ldots, S_s)^T$ where $s = \log^{O(1)}d$, and each $S_i \in \mathbb{R}^{(\tilde{O}(d)/s) \times n}$ with one $\pm 1$ on each column in a random row. Let $D$ is a diagonal matrix with $1/u_1^{1/p}, \ldots, 1/u_n^{1/p}$ on the diagonal. Let $\Pi = S D \in \mathbb{R}^{\tilde{O}(d) \times n}$. Note that the change of parameters of the OSE will not affect the contraction.

\begin{theorem}
\label{thm:ell_1}
Let $A$ be an Auerbach basis of a $d$-dimensional subspace of $(\mathbb{R}^n, \norm{\cdot}_1)$. Let $\Pi$ be defined as above. With probability $2/3$,
\begin{equation*}
\label{eq:w-1}
\Omega\left(\max\left\{1 / (d \log d), 1/\sqrt{d \log d \log n}\right\}\right)  \cdot \norm{Ax}_1 \le \norm{\Pi A x}_1 \le \tilde{O}(d) \cdot \norm{Ax}_1, \quad \forall x \in \mathbb{R}^d.
\end{equation*}
\end{theorem}

Same as Remark~\ref{rem:auerbach}, we can replace the Auerbach basis $A$ by any matrix $M$ whose column space is a $d$-dimensional subspace of $(\mathbb{R}^n, \norm{\cdot}_p)$. The embedding $\Pi M$ can be computed in time $O(\nnz(M) + \tilde{O}(d^{2}))$.

We need Khintchine's inequality.
\begin{fact}\label{fact:khintchine}
Let $z = \{z_1, \ldots, z_r\}$. Let $Z = \sum_{i=1}^r \sigma_i z_i$ for i.i.d.\ random variables $\sigma_i$
uniform in $\{-1,+1\}$. 
There exists a constant $c > 0$ for which for all $t > 0$
$$\Pr[|Z| > t \norm{z}_2] \leq \exp(-ct^2).$$
\end{fact}

Let $A = (A_1, \ldots, A_d)$ be an Auerbach basis of a $d$-dimensional subspace $(\mathbb{R}^n, \norm{\cdot}_1)$.
Applying Fact \ref{fact:khintchine} to a fixed entry $j$ of $SDA_i$ for a
fixed $i$, and letting $z^{i,j}$ denote the vector such that $(z^{i,j})_k = (A_i)_k$ if $S_{j,k} \neq 0$, and $(z^{i,j})_k = 0$ otherwise, 
we have for a constant $c' > 0$, 
$$\Pr\left[|(SDA_i)_j| > s \cdot c' \sqrt{\log d} \norm{Dz^{i,j}}_2\right] \leq \frac{1}{d^3}.$$
By a union bound, with probability $1-\frac{d^2 \log^{O(1)}d}{d^3}
= 1- \frac{\log^{O(1)}d}{d}$, for all $j$ and $i$
$$|(SDA_i)_j| \leq s \cdot c' \sqrt{\log d} \norm{Dz^{i,j}}_2,$$
which we denote by event $\mathcal{E}_9$ and condition on. 
%Notice that the probability is taken only over the choice of $S$, and therefore conditions only the $S$.

We define event $\mathcal{F}_{i,j}$ to be the event that 
\begin{equation}
\label{eq:m-1}
\norm{Dz^{i,j}}_2 \leq 100 c \cdot d^2 \log^{O(1)}d \norm{z^{i,j}}_1,
\end{equation}
where $c > 0$ is the constant of Corollary \ref{cor:upper}. 
By Corollary \ref{cor:upper},
$\Pr[\mathcal{F}_{i,j}] \geq 1 - {1}/\left({100 d^2 \log^{O(1)} d}\right)$. Let $\mathcal{F}_j = \bigwedge_{i \in [d]}\mathcal{F}_{i,j}$, and let $\mathcal{F} = \bigwedge_{j \in [d \log^{O(1)}d]} \mathcal{F}_j$. By union bounds, $\Pr[\mathcal{F}_j] \ge 1 - {1}/\left({100 d \log^{O(1)} d}\right)$ for all $j \in [d \log^{O(1)}d]$, and  $\Pr[\mathcal{F}] \ge 1 - {1}/{100} = {99}/{100}$.

\begin{claim}
\label{cla:sum}
$\E\left[\sum_{i \in [d], j \in [d  \log^{O(1)} d]} \norm{Dz^{i,j}}_2 \ |\ \mathcal{E}_9, \mathcal{F}\right]
 \leq  c_p \ln d \sum_{i \in [d]} \norm{A_i}_1$ for a constant $c_p > 0$.
\end{claim}
\begin{proof}
By independence, 
$$\E\left[\norm{Dz^{i,j}}_2 \ |\ \mathcal{E}_9, \mathcal{F}\right] = \E\left[\norm{Dz^{i,j}}_2 \ |\ \mathcal{E}_9, \mathcal{F}_j\right].$$

We now bound $\E[\norm{Dz^{i,j}}_2 \mid \mathcal{E}_9, \mathcal{F}_{i,j}]$. Letting $\eta = \Pr[\mathcal{E}_9 \wedge \mathcal{F}_{i,j}] \ge 99/100$, we have by Corollary~\ref{cor:upper} 
\begin{eqnarray*}
\E[\norm{Dz^{i,j}}_2 \mid \mathcal{E}_9, \mathcal{F}_{i,j}]
& = & \int_{r = 0}^{100 c d^2 \log^{O(1)}d} \Pr\left[\norm{Dz^{i,j}}_2 \geq r \norm{z^{i,j}}_1 
\mid \mathcal{E}_9, \mathcal{F}_{i,j}\right] dr \\
& \leq & \frac{1}{\eta} \left(1 + \int_{r = 1}^{100 c d^2 \log^{O(1)}d} \frac{c}{r}\right) dr\\
& \leq & c_p/2  \cdot \ln d \quad \quad \text{(for a large enough constant $c_p$)}.
\end{eqnarray*}
We can perform the following manipulation.
\begin{eqnarray*}
c_p/2 \cdot \ln d & \ge & \E\left[\norm{Dz^{i,j}}_2 \mid \mathcal{E}_9, \mathcal{F}_{i,j}\right] \\
&\ge& \E\left[\norm{Dz^{i,j}}_2 \mid \mathcal{E}_9, \mathcal{F}_j\right] \cdot \Pr[\mathcal{F}_j\ |\ \mathcal{F}_{i,j}] \\
&=& \E\left[\norm{Dz^{i,j}}_2 \mid \mathcal{E}_9, \mathcal{F}_j\right] \cdot \Pr[\mathcal{F}_j] / \Pr[\mathcal{F}_{i,j}] \\
&\ge& 1/2 \cdot \E\left[\norm{Dz^{i,j}}_2 \mid \mathcal{E}_9, \mathcal{F}_j\right] \\
&=& 1/2 \cdot \E\left[\norm{Dz^{i,j}}_2 \mid \mathcal{E}_9, \mathcal{F}\right]
\end{eqnarray*}
Therefore by linearity of expectation, $\E\left[\sum_{i \in [d], j \in [d  \log^{O(1)} d]} \norm{Dz^{i,j}}_2 \ |\ \mathcal{E}_9, \mathcal{F}\right]
 \leq  c_p \ln d \sum_{i \in [d]} \norm{A_i}_1$.
\end{proof}

Let $\mathcal{G}$ be the event that $\sum_{i \in [d], j \in [d  \log^{O(1)} d]} \norm{Dz^{i,j}}_2
 \leq  100 c_p \ln d \sum_{i \in [d]} \norm{A_i}_1$ conditioned on $\mathcal{E}_9, \mathcal{F}$. By Claim~\ref{cla:sum}, $\mathcal{G}$ holds with probability at least $99/100$. 
Then conditioned on $\mathcal{E}_9 \wedge \mathcal{F} \wedge \mathcal{G}$, which holds with probability at least $9/10$,
\begin{eqnarray*}
\norm{SDAx}_1 & \leq & \norm{x}_{\infty} \sum_{i \in [d]} \norm{SDA_i}_1\\
& \leq & \norm{Ax}_1 \sum_{i \in [d]} \norm{SDA_i}_1\\
& \le & \norm{Ax}_1 \sum_{i \in [d]} \sum_{j \in [d \log^{O(1)} d]} s \cdot c' \sqrt{\log d} \norm{D z^{i,j}}_2 \\
& \leq & \norm{Ax}_1 \cdot s \cdot c' \sqrt{\log d} \cdot 100 c_p \ln d \sum_{i \in [d]} \norm{A_i}_1\\
& \leq & \tilde{O}(d) \norm{Ax}_1,
\end{eqnarray*}
where the first inequality follows from the triangle inequality,
the second inequality uses that $\norm{x}_{\infty} \leq \norm{Ax}_1$ for a $(d,1,1)$-well-conditioned basis $A$, the third inequality uses 
Claim \ref{cla:sum}, and in the fourth inequality $\norm{A_i} = 1$ for all $i \in [d]$ for a $(d,1,1)$-well-conditioned basis $A$, and $s = \log^{O(1)}d$.

\subsection{A Tight Example}
\label{sec:app-tight}
We have the following example showing that given our embedding matrix $S D$, the distortion we get for $p=1$ is tight up to a polylog factor. The worst case $M$ is the same as the ``bad" example given in \cite{MM12}, that is, $M = (I_d, \mathbf{0})^T$ where $I_d$ is the $d \times d$ identity matrix. Suppose that the top $d$ rows of $M$ get perfectly hashed by $S$, then $\norm{S D M x}_2 = \left(\sum_{i \in [d]} (x_i / u_i)^2\right)^{1/2}$, where $u_i$ are i.i.d.\ exponentials. Let $i^* = \arg\max_{i \in [d]} 1/u_i$. We know from Property~\ref{prop:exp} that with constant probability,  $1/u_{i^*} = \Omega(d)$. Now if we choose $x$ such that $x_{i^*} = 1$ and $x_i= 0$ for all $i \neq i^*$, then $\norm{S D M x}_2 = d$.  On the other hand, we know that with constant probability, for $\Omega(d)$ of $i \in [d]$ we have $1/u_i = \Theta(1)$. Let $K\ (\abs{K} = \Omega(d))$ denote this set of indices. Now if we choose $x$ such that $x_i = 1/\abs{K}$ for all $i \in K$ and $x_i = 0$ for all $i \in [d] \backslash \abs{K}$, then  $\norm{S D M x}_2 = 1/\sqrt{\abs{K}} = O(1/\sqrt{d})$. Therefore the distortion is at least $\Omega(d^{3/2})$.

%% file: improve-L1.tex
\subsection{An Improved Contraction for $\ell_p\ (p \in [1,2))$ Subspace Embeddings when $d \ge \log^{2/p-1} n$}
\label{sec:improve-small-p}

%The analysis for the upper bound is the same as that in Section~\ref{sec:small-p-no-under}. 
In this section we give an improved analysis for the contraction assuming that $d \ge \log^{2/p-1} n$. 

Given a $y$, let $y_X \ (X \subseteq [n])$ be a vector such that $(y_X)_i = y_i$ if $i \in X$ and $0$ if $i \in [n] \backslash X$. For convenience, we assume that the coordinates of $y$ are sorted, that is, $y_1 \ge y_2 \ge \ldots \ge y_n$. Of course this order is unknown and not used by our algorithms.

We partition the $n$ coordinates of $y$ into $L = \log n + 2$ groups $W_1, \ldots, W_L$ such that $W_\ell = \{i\ |\ \norm{y}_p/2^\ell < y_i \le \norm{y}_p/2^{\ell-1}\}$. Let $w_\ell = \abs{W_\ell}\ (\ell \in [L])$ and let $W =\bigcup_{\ell \in [L]} W_\ell$. Thus 
$$\norm{y_W}_p^p \ge \norm{y}_p^p - n \cdot \norm{y}_p^p/(2^{L-1})^p \ge \norm{y}_p^p/2.$$ Let $K = c_K d \log d$ for a sufficiently large constant $c_K$. Define $T =  \{1, \ldots, K\}$ and $B = W \backslash T$. Obviously, $W_1 \cup \ldots \cup W_{\log K - 1} \subseteq T$. Let $\lambda = 1/(10 d^p K)$ be a threshold parameter.

As before (Section~\ref{sec:small-p-no-under}), we have $\norm{SDy}_2 \ge 1/2 \cdot \norm{Dy}_2$. Now we analyze $\norm{Dy}_2$ by two cases.

\paragraph{Case 1: $\norm{y_T}_p^p \ge \norm{y}_p^p/4$.}  

Let $H = \{i\ |\ (i \in [n]) \wedge (\ell_i^p \ge \lambda)\}$, where $\ell_i^p$ is the $i$-th leverage score of $A$. Since $\sum_{i \in [n]} \ell_i^p = d$, it holds that $\abs{H} \le d / \lambda$.

We next claim that $\norm{y_{T \cap H}}_p^p \ge \norm{y}_p^p / 8$. To see this, recall that for each $y_i\ (i \in [n])$ we have $\abs{y_i^p} \le d^{p-1} \ell_i^p$ (Property~\ref{prop:y}). Suppose that $\norm{y_{T \cap H}}_p^p \le \norm{y}_p^p / 8$, let $y_{i_{\max}}$ be the  coordinate in $y_{T \backslash H}$ with maximum absolute value, then 
\begin{eqnarray*}
\abs{y_{i_{\max}}^p} &\ge& \norm{y}_p^p / (8K) \\
&\ge& (1/d) / (8K) \quad \text{(by Property~\ref{prop:y})}\\
&>& d^{p-1} \lambda \\
&>& d^{p-1} \ell_{i_{\max}}^p.  \quad ({i_{\max}} \not\in H)
\end{eqnarray*} 
This is a contradiction. 

Now we consider $\{u_i\ |\ i \in H\}$. Since the CDF of an exponential $u$ is $(1 - e^{-x})$, we have with probability $(1 - d^{-10})$ that $1/u \ge 1/(10 \log d)$. By a union bound, with probability $(1 - d^{-10} \abs{H}) \ge (1 - d^{-10} \cdot 10 d^{p+1} K) \ge 0.99$, it holds that $1/u_i \ge 1/(10 \log d)$ for all $i \in H$. Let $\mathcal{E}_7$ be this event which we condition on. Then for any $y$ such that $\norm{y_T}_p^p \ge \norm{y}_p^p/4$, we have $\sum_{i \in T \cap H} \abs{y_i^p} /u_i \ge \norm{y}_p^p / (80 \log d)$, and consequently, 
$$\norm{Dy}_2 \ge \frac{\norm{Dy}_p}{K^{1/p-1/2}} \ge \frac{\norm{y}_p}{(80 \log d)^{1/p} \cdot K^{1/p-1/2}}.$$

\paragraph{Case 2: $\norm{y_B}_p^p \ge \norm{y}_p^p/4$.} Let $W'_\ell = B \cap W_\ell\ (\ell \in [L])$ and $w'_\ell = \abs{W'_\ell}$. Let $F = \{\ell\ |\ w'_\ell \ge K/32\}$ and let $W' = \bigcup_{\ell \in F} W_\ell$. We have
\begin{eqnarray*}
\norm{y_{W'}}_p^p &\ge& \norm{y}_p^p/4 - \sum_{\ell = \log K}^L \left(K/32 \cdot  (\norm{y}_p / 2^{\ell - 1})^p \right) \\
&\ge&  \norm{y}_p^p/4 -  \norm{y}_p^p \cdot K/32 \cdot \sum_{\ell = \log K}^L \left(1 / 2^{\ell-1} \right) \\
&\ge& \norm{y}_p^p/8.
\end{eqnarray*}

For each $\ell \in F$, let $\alpha_\ell = w'_\ell / (2^\ell)^p$. We have 
$$\norm{y}_p^p/8 \le \norm{y_{W'}}_p^p = \sum_{\ell \in F} \left(w'_\ell \cdot \left({\norm{y}_p}/{2^{\ell-1}}\right)^p \right) \le \sum_{\ell \in F} \left( \alpha_\ell \cdot 4 \norm{y}_p^p \right).$$
Thus $\sum_{\ell \in F} \alpha_\ell \ge 1/32$.

Now for each $\ell \in F$, we consider $\sum_{i \in W_{\ell}} \left(y_i / u_i^{1/p}\right)^2$.  By Property $\ref{prop:exp}$, for an exponential $u$ we have $\Pr[1 / u \ge 
w'_{\ell}/K] \ge c'_e \cdot K/w'_{\ell} \ (c'_e = \Theta(1))$. By a Chernoff bound, with probability $(1 - 
e^{-\Omega(K)})$, there are at least $\Omega(K)$ of $i \in W_{\ell}$ such 
that $1/u_i \ge w'_{\ell} / K$. Thus with probability at least $(1 - 
e^{-\Omega(K)})$, we have
\begin{equation*}
\sum_{i \in W_{\ell}} \left(y_i / u_i^{1/p}\right)^2 \ge \Omega(K)  \cdot \left(\frac{\norm{y}_p}{2^{\ell}} \cdot \frac{{w'_{\ell}}^{1/p}}{K^{1/p}}\right)^2 \ge \Omega\left(\frac{\alpha_\ell^{2/p} \norm{y}_p^{2}}{K^{2/p - 1}}\right).
\end{equation*}
Therefore with probability $(1 - L \cdot e^{-\Omega(K)}) \ge (1 - 
e^{-\Omega(d \log d)})$, we have 
\begin{eqnarray}
\norm{Dy}_2^2 &\ge& \sum_{\ell \in F}\sum_{i \in W_{\ell}} \left(y_i / u_i^{1/p}\right)^2 \nonumber \\
&\ge& \Omega\left(\frac{\norm{y}_p^{2}}{K^{2/p - 1}} \cdot \sum_{\ell \in F} \alpha_\ell^{2/p} \right) \nonumber \\
&\ge& \Omega \left(\frac{\norm{y}_p^{2}}{(K \log n)^{2/p - 1}}\right) \label{eq:z-3} \quad \textstyle (\sum_{\ell \in F} \alpha_\ell \ge 1/32 \text{ and } \abs{F} \le \log n) 
\end{eqnarray}

Since the success probability is as high as $(1 - 
e^{-\Omega(d \log d)})$, we can further show that (\ref{eq:z-3}) holds for all $y \in \text{range}(A)$ using a net-argument as in previous sections.

To sum up the two cases, we have that for $\forall y \in \text{range}(A)$ and $p \in [1,2)$, $\norm{Dy}_2 \ge \Omega\left(\frac{\norm{y}_p}{(d \log d \log n)^{\frac{1}{p}-\frac{1}{2}}}\right)$.

%% file: regression.tex
\section{Regression}
\label{sec:regression}
We need the following lemmas for $\ell_p$ regression.

\begin{lemma}[\cite{CDMMMW12}]
\label{lem:compute-well-condition}
Given a matrix $M \in \mathbb{R}^{n \times d}$  with full column rank and $p \in [1, \infty)$, it takes at most $O(n d^3 \log n)$ time to find a matrix $R \in \mathbb{R}^{d \times d}$ such that $MR^{-1}$ is $(\alpha, \beta, p)$-well-conditioned with $\alpha\beta \le 2d^{1+\max\{1/2, 1/p\}}$.
\end{lemma}

\begin{lemma}[\cite{CDMMMW12}]
\label{lem:subspace-sampling}
Given a matrix $M \in \mathbb{R}^{n \times d}, p \in [1, \infty), \eps > 0$, and a matrix $R \in \mathbb{R}^{d \times d}$ such that $M R^{-1}$ is $(\alpha, \beta, p)$-well-conditioned, it takes $O(\text{nnz}(M) \cdot \log n)$  time to compute a sampling matrix $\Pi \in \mathbb{R}^{t \times n}$ such that with probability $0.99$,
$(1 - \eps)\norm{M x}_p \le \norm{\Pi M x}_p \le (1+\eps)\norm{M x}_p, \ \forall x \in \mathbb{R}^d.$
The value $t$ is $O\left((\alpha \beta)^p d \log(1/\eps) / \eps^2\right)$ for $1 \le p < 2$ and $O\left((\alpha \beta)^p d^{p/2} \log(1/\eps) / \eps^2\right)$ for $p > 2$.
\end{lemma}

\begin{lemma}[\cite{CDMMMW12}]
\label{lem:correct}
Given an $\ell_p$-regression problem specified by $M \in \mathbb{R}^{n \times (d - 1)}, b \in \mathbb{R}^n$, and $p \in [1, \infty)$, let $\Pi$ be a $(1\pm \eps)$-distortion embedding matrix of the subspace spanned by $M$'s columns and $b$ from Lemma~\ref{lem:subspace-sampling}, and let $\hat{x}$ be an optimal solution to the sub-sampled problem $\min_{x \in \mathbb{R}^d} \norm{\Pi M x - \Pi b}_p$. Then $\hat{x}$ is a $\frac{1+\eps}{1-\eps}$-approximation solution to the original problem.
\end{lemma}

\subsection{Regression for $p$-norm with $p > 2$}
\begin{lemma}
\label{lem:big-p-well-condition}
Let $\Pi \in \mathbb{R}^{m \times n}$ be a subspace embedding matrix of the $d$-dimensional normed space spanned by the columns of matrix $M \in \mathbb{R}^{n \times d}$ such that 
$\mu_1  \norm{Mx}_p \le \norm{\Pi M x}_\infty \le \mu_2  \norm{Mx}_p$ for $\forall x \in \mathbb{R}^d.$
If $R$ is a matrix such that $\Pi M R^{-1}$ is $(\alpha, \beta, \infty)$-well-conditioned, then $MR^{-1}$ is $(\beta \mu_2 , d^{1/p} \alpha/\mu_1, p)$-well-conditioned for any $p \in (2, \infty)$.
\end{lemma}
\begin{proof}
According to Definition~\ref{def:well-condition-1}, we only need to prove
\begin{eqnarray}
\norm{x}_q & \le & \norm{x}_1 
 \le  \beta \norm{\Pi M R^{-1} x}_\infty  \quad (\Pi M R^{-1} \text{ is $(\alpha, \beta, \infty)$-well-conditioned}) \nonumber \\
&\le& \beta \cdot \mu_2 \norm{M R^{-1} x}_p. \quad \text{(property of $\Pi$)}  \nonumber \quad \\
\text{And,}\quad \quad \quad  \quad \quad && \nonumber \\
\norm{MR^{-1}}_p^p &=& \sum_{i \in [d]} \norm{M R^{-1} e_i}_p^p \quad \text{($e_i$ is the standard basis in $\mathbb{R}^d$)} \nonumber \\
&\le& 1/\mu_1^p \sum_{i \in [d]} \norm{\Pi M R^{-1} e_i}_\infty^p \quad \text{(property of $\Pi$)}  \nonumber \\
&\le& 1 /\mu_1^p \cdot d \alpha^p. \quad (\Pi M R^{-1} \text{ is $(\alpha, \beta, \infty)$-well-conditioned}) \nonumber
\end{eqnarray}
\end{proof}

\begin{theorem}
\label{thm:regression-big-p}
There exists an algorithm that given an $\ell_p$-regression problem specified by $M \in \mathbb{R}^{n \times (d-1)}, b \in \mathbb{R}^n$ and $p \in (2, \infty)$, with constant probability computes a $(1+\eps)$-approximation to an $\ell_p$-regression problem in time $\tilde{O}\left(\text{nnz}(\bar{M}) + n^{1-2/p} d^{4+2/p} + d^{8+4p} + \phi(\tilde{O}(d^{3+2p}/\eps^2), d)\right)$, where $\bar{M} = [M, -b]$ and $\phi(t, d)$ is the time to solve $\ell_p$-regression problem on $t$ vectors in $d$ dimensions.
\end{theorem}
\begin{proof}
Our algorithm is similar to those $\ell_p$-regression algorithms described in \cite{DDHKM09,CDMMMW12,MM12}. For completeness we sketch it here. Let $\Pi$ be the subspace embedding matrix in Section~\ref{sec:big} for $p > 2$. By Theorem~\ref{thm:big-p}, we have $(\mu_1, \mu_2) = \left(\Omega(1/(d\log d)^{1/p}), O((d \log d)^{1/p})\right)$.

\paragraph{Algorithm: $\ell_p$ regression for $p>2$}
\label{alg:regression-big-p}
\begin{enumerate}
\item Compute $\Pi \bar{M}$. 

\item Use Lemma~\ref{lem:compute-well-condition} to compute a matrix $R  \in \mathbb{R}^{d \times d}$ such that $\Pi \bar{M} R^{-1}$ is $(\alpha, \beta, \infty)$-well-conditioned with $\alpha\beta \le 2d^{3/2}$. By Lemma~\ref{lem:big-p-well-condition}, $\bar{M} R^{-1}$ is $(\beta \mu_2, d^{1/p}\alpha/\mu_1, p)$-well-conditioned. 

\item Given $R$, use Lemma~\ref{lem:subspace-sampling} to find a sampling matrix $\Pi^1$ such that \\
$(1-\eps) \cdot \norm{\bar{M} x}_p \le \norm{\Pi^1 \bar{M} x}_p \le (1+\eps) \cdot \norm{\bar{M} x}_p, \quad \forall x \in \mathbb{R}^d.$

\item Compute $\hat{x}$ which is the optimal solution to the sub-sampled problem $\min_{x \in \mathbb{R}^d} \norm{\Pi^1 M x - \Pi^1 b}_p$.
\end{enumerate}

\paragraph{Analysis.}
The correctness of the algorithm is guaranteed by Lemma~\ref{lem:correct}.
Now we analyze the running time. Step $1$ costs time $O(\text{nnz}(\bar{M}))$, by our choice of $\Pi$. Step $2$ costs time $O(m d^3 \log m)$ by Lemma~\ref{lem:compute-well-condition}, where $m = O(n^{1 - 2/p} \log n (d \log d)^{1 + 2/p} + d^{5+4p})$. Step $3$ costs time $O(\text{nnz}(\bar{M}) \log n)$  by Lemma~\ref{lem:subspace-sampling}, giving a sampling matrix $\Pi^1 \in \mathbb{R}^{t \times n}$ with $t = O(d^{3+2p} \log^2 d \log(1/\eps)/\eps^2)$. Step $4$ costs time $\phi(t, d)$, which is the time to solve $\ell_p$-regression problem on $t$ vectors in $d$ dimensions. 
To sum up, the total running time is
$O\left(\text{nnz}(\bar{M}) \log n + n^{1-2/p} d^{4+2/p} \log^2 n \log^{1+2/p} d + d^{8+4p} \log n + \phi(O(d^{3+2p} \log^2 d \log(1/\eps)/\eps^2), d)\right).$
\end{proof}

\subsection{Regression for $p$-norm with $1 \le p < 2$}
\begin{theorem}
\label{thm:regression-small-p}
There exists an algorithm that given an $\ell_p$ regression problem specified by $M \in \mathbb{R}^{n \times (d-1)}, b \in \mathbb{R}^n$ and $p \in [1, 2)$, with constant probability computes a $(1+\eps)$-approximation to an $\ell_p$-regression problem in time $\tilde{O}\left(\text{nnz}(\bar{M}) + d^{7-p/2} + \phi(\tilde{O}(d^{2+p}/\eps^2), d)\right)$, where $\bar{M} = [M, -b]$ and $\phi(t, d)$ is the time to solve $\ell_p$-regression problem on $t$ vectors in $d$ dimensions.
\end{theorem}

We first introduce a few lemmas.

\begin{lemma}[\cite{SW11,MM12}]
\label{lem:dense-embed}
Given $M \in \mathbb{R}^{n \times d}$ with full column rank, $p \in [1, 2)$, and $\Pi \in \mathbb{R}^{m \times n}$ whose entries are i.i.d. $p$-stables, if $m = c d\log d$ for a sufficiently large constant $c$, then with  probability $0.99$, we have 
$$\Omega(1) \cdot \norm{M x}_p \le \norm{\Pi M x}_p \le O((d \log d)^{1/p}) \cdot \norm{M x}_p, \quad \forall x \in \mathbb{R}^d.$$
In addition, $\Pi M$ can be computed in time $O(nd^{\omega-1})$ where $\omega$ is the exponent of matrix multiplication.
\end{lemma}

\begin{lemma}
\label{lem:small-p-well-condition}
Let $\Pi \in \mathbb{R}^{m \times n}$ be a subspace embedding matrix of the $d$-dimensional normed space  spanned by the columns of matrix $M \in \mathbb{R}^{n \times d}$ such that 
\begin{equation}
\label{eq:l-3}
\mu_1 \cdot \norm{Mx}_p \le \norm{\Pi M x}_2 \le \mu_2 \cdot \norm{Mx}_p, \quad \forall x \in \mathbb{R}^d.
\end{equation}
If $R$ is the ``$R$" matrix in the $QR$-decomposition of $\Pi M$, then $MR^{-1}$ is $(\alpha, \beta, p)$-well-conditioned with $\alpha \beta \le d^{1/p} \mu_2/\mu_1$ for any $p \in [1, 2)$.
\end{lemma}
\begin{proof}
We first analyze $\Delta_p(M R^{-1}) = \mu_2 / \mu_1$ (Definition \ref{def:well-condition-2}).
\begin{eqnarray}
\norm{MR^{-1}x}_p & \le & 1/\mu_1 \cdot \norm{\Pi M R^{-1} x}_2 \quad \text{(by (\ref{eq:l-3}))}  \nonumber \\
& = & 1/\mu_1 \cdot \norm{Q x}_2 \quad (\Pi M R^{-1} = Q R R^{-1} = Q)  \nonumber \\
& = & 1/\mu_1 \cdot \norm{x}_2  \quad \text{($Q$ has orthonormal columns)} \nonumber 
\end{eqnarray}
And
\begin{eqnarray}
\norm{MR^{-1}x}_p & \ge & 1/\mu_2 \cdot \norm{\Pi M R^{-1} x}_2 \quad \text{(by (\ref{eq:l-3}))}  \nonumber \\
& = & 1/\mu_2 \cdot \norm{Q x}_2   \nonumber \\
& = & 1/\mu_2 \cdot \norm{x}_2 \nonumber 
\end{eqnarray}
Then by Lemma~\ref{lem:relation-well-condition} it holds that 
$$\alpha \beta = \Delta'_p(M R^{-1}) \le d^{\max\{1/2, 1/p\}} \Delta_p(M R^{-1}) = d^{1/p} \mu_2/\mu_1.$$
\end{proof}

%\subsection{Proof for Theorem~\ref{thm:regression-small-p}}
%\label{sec:app-4-1}
\begin{proof} (for Theorem~\ref{thm:regression-small-p})
The regression algorithm for $1 \le p < 2$ is similar but slightly more complicated than that for $p > 2$, since we are trying to optimize the dependence on $d$ in the running time. Let $\Pi$ be the subspace embedding matrix in Section~\ref{sec:small} for $1 \le p < 2$. By theorem~\ref{thm:small-p}, we have $(\mu_1, \mu_2) = (\Omega(1/(d\log d)^{1/p}), O((d \log d)^{1/p}))$ (we can also use $(\Omega(1 / {(d \log d \log n)^{\frac{1}{p}-\frac{1}{2}}}), O((d \log d)^{1/p}))$ which will give the same result).

\paragraph{Algorithm: $\ell_p$-Regression for $1\le p<2$}
\label{alg:regression-small-p}

\begin{enumerate}
\item Compute $\Pi \bar{M}$. 

\item Compute the $QR$-decomposition of $\Pi \bar{M}$. Let $R  \in \mathbb{R}^{d \times d}$ be the ``$R$" in the $QR$-decomposition.

\item Given $R$, use Lemma~\ref{lem:subspace-sampling} to find a sampling matrix $\Pi^1 \in \mathbb{R}^{t_1 \times n}$ such that 
\begin{equation}
\label{eq:algo2-1}
(1-1/2) \cdot \norm{\bar{M} x}_p \le \norm{\Pi^1 \bar{M} x}_p \le (1+1/2) \cdot \norm{\bar{M} x}_p, \quad \forall x \in \mathbb{R}^d.
\end{equation}

\item Use Lemma~\ref{lem:dense-embed} to compute a matrix $\Pi^2 \in \mathbb{R}^{t_2 \times t_1}$ for $\Pi^1 \bar{M}$ such that 
$$\Omega(1) \cdot \norm{\Pi^1 \bar{M} x}_p \le \norm{\Pi^2 \Pi^1 \bar{M} x}_p \le O((d \log d)^{1/p}) \cdot \norm{\Pi^1 \bar{M} x}_p, \quad \forall x \in \mathbb{R}^d.$$
Let $\Pi^3 = \Pi^2 \Pi^1 \in \mathbb{R}^{t_2 \times n}$. By (\ref{eq:algo2-1}) and $\norm{z}_2 \le \norm{z}_p \le m^{1/p-1/2} \norm{z}_2$ for any $z \in \mathbb{R}^m$, we have
$$\Omega(1/{t_2}^{1/p-1/2}) \cdot \norm{\bar{M} x}_p \le \norm{\Pi^3 \bar{M} x}_2 \le O((d \log d)^{1/p}) \cdot \norm{\bar{M} x}_p, \quad \forall x \in \mathbb{R}^d.$$

\item Compute the $QR$-decomposition of $\Pi^3 \bar{M}$. Let $R_1  \in \mathbb{R}^{d \times d}$ be the ``$R$" in the $QR$-decomposition.

\item Given $R_1$, use Lemma~\ref{lem:subspace-sampling} again to find a sampling matrix $\Pi^4  \in \mathbb{R}^{t_3 \times n}$ such that $\Pi^4$ is a $(1\pm 1/2)$-distortion embedding matrix of the subspace spanned by $\bar{M}$. 

\item Use Lemma~\ref{lem:compute-well-condition} to compute a matrix $R_2  \in \mathbb{R}^{d \times d}$ such that $\Pi^4 \bar{M} {R_2}^{-1}$ is $(\alpha, \beta, p)$-well-conditioned with $\alpha\beta \le 2d^{1+1/p}$. 

\item Given $R_2$, use Lemma~\ref{lem:subspace-sampling} again to find a sampling matrix $\Pi^5 \in \mathbb{R}^{t_4 \times n}$ such that $\Pi^5$ is a $(1\pm \eps)$-distortion embedding matrix of the subspace spanned by $\bar{M}$. 

\item Compute $\hat{x}$ which is the optimal solution to the sub-sampled problem $\min_{x \in \mathbb{R}^d} \norm{\Pi^5 M x - \Pi^5 b}_p$.
\end{enumerate}

\paragraph{Analysis.}
The correctness of the algorithm is guaranteed by Lemma~\ref{lem:correct}.
Now we analyze the running time. Step $1$ costs time $O(\text{nnz}(\bar{M}))$, by our choice of $\Pi$. Step $2$ costs time $O(m d^2) = O(d^{3+\gamma})$ using standard $QR$-decomposition, where $\gamma$ is an arbitrarily small constant. Step $3$ costs time $O(\text{nnz}(\bar{M}) \log n)$  by Lemma~\ref{lem:subspace-sampling}, giving a sampling matrix $\Pi^1 \in \mathbb{R}^{t_1 \times n}$ with $t_1 = O(d^{4} \log^2 d)$. Step $4$ costs time $O(t_1 d^{\omega-1}) = O(d^{3+\omega} \log^2 d)$ where $\omega$ is the exponent of matrix multiplication, giving a  matrix $\Pi^3 \in \mathbb{R}^{t_2 \times n}$ with $t_2 = O(d\log d)$. Step $5$ costs time $O(t_2 d^2) = O(d^3 \log d)$. Step $6$ costs time $O(\text{nnz}(\bar{M}) \log n)$  by Lemma~\ref{lem:subspace-sampling}, giving a sampling matrix $\Pi^4 \in \mathbb{R}^{t_3 \times n}$ with $t_3 = O(d^{4-p/2} \log^{2-p/2} d)$. Step $7$ costs time $O(t_3 d^3 \log t_3) = O(d^{7-p/2} \log^{3-p/2} d)$. Step $8$ costs time $O(\text{nnz}(\bar{M}) \log n)$  by Lemma~\ref{lem:subspace-sampling}, giving a sampling matrix $\Pi^5 \in \mathbb{R}^{t_4 \times n}$ with $t_4 = O(d^{2+p} \log(1/\eps) / \eps^2)$. Step $9$ costs time $\phi(t_4, d)$, which is the time to solve $\ell_p$-regression problem on $t_4$ vectors in $d$ dimensions. 
To sum up, the total running time is $$O\left(\text{nnz}(\bar{M}) \log n + d^{7-p/2} \log^{3-p/2} d + \phi(O(d^{2+p} \log(1/\eps)/\eps^2), d)\right).$$
\end{proof}

\begin{remark}
In \cite{MM12} an algorithm together with several variants for $\ell_1$-regression are proposed, all with running time of the form $\tilde{O}\left(\text{nnz}(\bar{M}) + \text{poly}(d) + \phi(\tilde{O}(\text{poly}(d)/\eps^2), d)\right)$. Among all these variants, the power of $d$ in $\text{poly}(d)$ (ignoring log factors) in the second term is at least $7$, and the power of $d$ in $\text{poly}(d)$ in the third term is at least $3.5$. In our algorithm both terms are improved.
\end{remark}

\paragraph{Application to $\ell_1$ Subspace Approximation.} 
Given a matrix $M \in \mathbb{R}^{n \times d}$ and a parameter $k$, the $\ell_1$-subspace approximation is to compute a matrix $\hat{M}$ of rank $k \in [d-1]$ such that $\norm{M - \hat{M}}_1$ is minimized. When $k = d - 1$, $\hat{M}$ is a hyperplane, and the problem is called {\em $\ell_1$ best hyperplane fitting}. In \cite{CDMMMW12} it is shown that this problem is equivalent to solving the regression problem $\min_{W \in \mathcal{C}} \norm{AW}_1$, where the constraint set is $\mathcal{C} = \{W \in \mathbb{R}^{d \times d} : W_{ii} = -1 \}$. Therefore, our $\ell_1$-regression result directly implies an improved algorithm for $\ell_1$ best hyperplane fitting. Formally, we have 
\begin{theorem}
Given $M \in \mathbb{R}^{n \times d}$, there exists an algorithm that computes a $(1+\eps)$-approximation to the $\ell_1$ best hyperplane fitting problem with probability $0.9$, using time $O\left(\nnz(M) \log n + \frac{1}{\eps^2} \poly(d, \log\frac{d}{\eps}) \right)$.
\end{theorem}
The $\poly(d)$ factor in our algorithm is better than those by using the regression results in \cite{CW12,CDMMMW12,MM12}.

%% file: communication.tex
\section{Regression in the Distributed Setting}
\label{sec:dist-regression}
In this section we consider the $\ell_p$-regression problem in the distributed setting, where we have $k$ machines $P_1, \ldots, P_k$ and one central server. Each machine has a disjoint subset of the rows of $M \in \mathbb{R}^{n \times (d-1)}$ and $b \in \mathbb{R}^d$. The server has a $2$-way communication channel with each machine, and the server wants to communicate with the $k$ machines to solve the $\ell_p$-regression problem specified by $M, b$ and $p$. Our goal is to minimize the overall communication of the system, as well as the total running time.

Let $\bar{M} = [M, -b]$. Let $I_1, \ldots, I_k$ be the sets of rows that $P_1, \ldots, P_k$ have, respectively. Let $\bar{M}_i\ (i \in [k])$ be the matrix by setting all rows $j \in [n] \backslash I_i$ in $\bar{M}$ to $0$. We use $\Pi$ to denote the subspace embedding matrix proposed in Section~\ref{sec:big} for $p>2$ and Section~\ref{sec:small} for $1 \le p < 2$, respectively. We assume that both the server and the $k$ machines agree on such a $\Pi$ at the beginning of the distributed algorithms using, for example, shared randomness.

\subsection{Distributed $\ell_p$-regression for $p > 2$}
The distributed algorithm for $\ell_p$ regression with $p > 2$ is just a distributed implementation of Algorithm~\ref{alg:regression-big-p}.  

\paragraph{Algorithm:  Distributed $\ell_p$-regression for $p>2$}
\label{alg:dist-regression-big-p}
\begin{enumerate}
\item Each machine computes and sends $\norm{\bar{M}_i}_p$ to the server. And then the server computes $\norm{\bar{M}}_p = \left(\sum_{i \in [k]} \norm{\bar{M}_i}_p^p\right)^{1/p}$  and sends to each site. $\norm{\bar{M}}_p$ is needed for  Lemma~\ref{lem:subspace-sampling} which we will use later.

\item Each machine $P_i$ computes and sends $\Pi \bar{M}_i$ to the server.

\item The server computes $\Pi \bar{M}$ by summing up $\Pi \bar{M}_i\ (i = 1, \ldots, k)$. Next, the server uses Lemma~\ref{lem:compute-well-condition} to compute a matrix $R \in \mathbb{R}^{d \times d}$ such that $\Pi \bar{M} R^{-1}$ is $(\alpha, \beta, \infty)$-well-conditioned with $\alpha\beta \le 2d^{3/2}$, and sends $R$ to each of the $k$ machines.

\item Given $R$ and $\norm{\bar{M}}_p$, each machine uses Lemma~\ref{lem:subspace-sampling} to compute a sampling matrix $\Pi^1_i$ such that $\Pi^1_i$ is a $(1\pm \eps)$-distortion embedding matrix of the subspace spanned by $\bar{M_i}$, and then sends the sampled rows of $\Pi^1_i \bar{M}
_i$ that are in $I_i$ to the server.

\item The server constructs a global matrix $\Pi^1 \bar{M}$ such that the $j$-th row of $\Pi^1 \bar{M}$ is just the $j$-th row of $\Pi^1_i \bar{M}_i$ if $(j \in I_i) \wedge (j \text{ get sampled})$, and $0$ otherwise. Next, the server computes $\hat{x}$ which is the optimal solution to the sub-sampled problem $\min_{x \in \mathbb{R}^d} \norm{\Pi^1 M x - \Pi^1 b}_p$.
\end{enumerate}

\paragraph{Analysis.}
Step $1$ costs communication $O(k)$. Step $2$ costs communication $O(k md)$ where \\
$m = O(n^{1-2/p} \log n  (d\log d)^{1+2/p} + d^{5+4p})$. Step $3$ costs communication $O(k d^2)$. Step $4$ costs communication $O(td+k)$ where $t = O(d^{3+2p} \log^2 d \log(1/\eps)/\eps^2)$, that is, the total number of rows get sampled in rows $I_1 \cup I_2 \cup \cdots \cup I_k$. 
Therefore the total communication cost is 
$$O\left(k n^{1-2/p} d^{2+2/p} \log n \log^{1+2/p} d + kd^{6+4p} + d^{4+2p} \log^2 d \log(1/\eps)/\eps^2\right).$$

The total running time of the system, which is essentially the running time of the centralized algorithm (Theorem~\ref{thm:regression-big-p}) plus the communication cost, is $$O\left(\nnz(\bar{M}) \log n + (k + d^2 \log n) (n^{1-2/p} d^{2+2/p} \log n \log^{1+2/p} d + d^{6+4p})+ \phi(O(d^{3+2p} \log^2 d \log(1/\eps)/\eps^2), d)\right).$$ %Note that $O(k d^2)$ and $O(d^{4+2p} \log^2 d \log(1/\eps)/\eps^2)$ are dominated by the second term and third term of the expression above.

\subsection{Distributed $\ell_p$-regression for $1 \le p < 2$}
The distributed algorithm for $\ell_p$-regression with $1 \le p < 2$ is  a distributed implementation of Algorithm~\ref{alg:regression-small-p}. 

\paragraph{Algorithm: Distributed $\ell_p$-regression for $1 \le p < 2$}
\label{alg:dist-regression-small-p}
\begin{enumerate}
\item Each machine computes and sends $\norm{\bar{M}_i}_p$ to the server. And then the server computes $\norm{\bar{M}}_p = \left(\sum_{i \in [k]} \norm{\bar{M}_i}_p^p\right)^{1/p}$  and sends to each site.

\item Each machine $P_i$ computes and sends $\Pi \bar{M}_i$ to the server.

\item The server computes $\Pi \bar{M}$ by summing up $\Pi \bar{M}_i\ (i = 1, \ldots, k)$. Next, the server computes a $QR$-decomposition of $\Pi  \bar{M}$, and sends $R$ (the ``$R$" in $QR$-decomposition) to each of the $k$ machines.

\item Given $R$ and $\norm{\bar{M}}_p$, each machine $P_i$ uses Lemma~\ref{lem:subspace-sampling} to compute a sampling matrix $\Pi_i^1 \in \mathbb{R}^{t_1 \times n}$ such that $\Pi_i^1$ is a $(1\pm 1/2)$-distortion embedding matrix of the subspace spanned by $\bar{M_i}$, and then sends the sampled rows of $\Pi_i^1 \bar{M}
_i$ that are in $I_i$ to the server.

\item The server constructs a global matrix $\Pi^1 \bar{M}$ such that the $j$-th row of $\Pi^1 \bar{M}$ is just the $j$-th row of $\Pi^1_i \bar{M}_i$ if $(j \in I_i) \wedge (j \text{ get sampled})$, and $0$ otherwise. After that, the server uses Lemma~\ref{lem:dense-embed} to compute a matrix $\Pi^2 \in \mathbb{R}^{t_2 \times t_1}$ for $\Pi^1 \bar{M}$. Next, the server computes a $QR$-decomposition of $\Pi^2 \Pi^1 \bar{M}$, and sends $R_1$ (the ``$R$" in $QR$-decomposition) to each of the $k$ machines.

\item Given $R_1$ and $\norm{\bar{M}}_p$, each machine $P_i$ uses Lemma~\ref{lem:subspace-sampling} again to compute a sampling matrix $\Pi_i^4 \in \mathbb{R}^{t_3 \times n}$ such that $\Pi_i^4$ is a $(1\pm 1/2)$-distortion embedding matrix of the subspace spanned by $\bar{M_i}$, and then sends the sampled rows of $\Pi_i^4 \bar{M}
_i$ that are in $I_i$ to the server.

\item The server constructs a global matrix $\Pi^4 \bar{M}$ such that the $j$-th row of $\Pi^4 \bar{M}$ is just the $j$-th row of $\Pi^4_i \bar{M}_i$ if $(j \in I_i) \wedge (j \text{ get sampled})$, and $0$ otherwise. Next,  the server uses Lemma~\ref{lem:compute-well-condition} to compute a matrix $R_2 \in \mathbb{R}^{d \times d}$ such that $\Pi \bar{M} {R_2}^{-1}$ is $(\alpha, \beta, p)$-well-conditioned with $\alpha\beta \le 2d^{1+1/p}$, and sends $R_2$ to each of the $k$ machines.

\item Given $R_2$ and $\norm{\bar{M}}_p$, each machine $P_i$ uses Lemma~\ref{lem:subspace-sampling} again to compute a sampling matrix $\Pi^5_i \in \mathbb{R}^{t_4 \times n}$ such that $\Pi^5_i$ is a $(1\pm \eps)$-distortion embedding matrix of the subspace spanned by $\bar{M_i}$, and then sends the sampled rows of $\Pi^5_i \bar{M}_i$ that are in $I_i$ to the server.

\item The server constructs a global matrix $\Pi^5 \bar{M}$ such that the $j$-th row of $\Pi^5 \bar{M}$ is just the $j$-th row of $\Pi_i^5 \bar{M}_i$ if $(j \in I_i) \wedge (j \text{ get sampled})$, and $0$ otherwise.  Next, the server computes $\hat{x}$ which is the optimal solution to the sub-sampled problem $\min_{x \in \mathbb{R}^d} \norm{\Pi^5 M x - \Pi^5 b}_p$.
\end{enumerate}

\paragraph{Communication and running time.}
Step $1$ costs communication $O(k)$.
Step $2$ costs communication $O(k md)$ where $m = O(d^{1+\gamma})$ for some arbitrarily small $\gamma$. Step $3$ costs communication $O(k d^2)$. Step $4$ costs communication $O(t_1d+k)$ where $t_1 = O(d^4 \log^2 d)$. Step $5$ costs communication $O(k d^2)$. Step $6$ costs communication $O(t_3d+k)$ where $t_3 = O(d \log d)$. Step $7$ costs communication $O(k d^2)$. Step $8$ costs communication $O(t_4d+k)$ where $t_4 = O(d^{2+p} \log(1/\eps)/\eps^2)$. Therefore the total communication cost is 
$$O\left(k d^{2+\gamma} + d^5 \log^2 d + d^{3+p} \log(1/\eps)/\eps^2\right).$$

The total running time of the system, which is essentially the running time of the centralized algorithm (Theorem~\ref{thm:regression-small-p}) plus the communication cost, is
$$O\left(\text{nnz}(\bar{M}) \log n + k d^{2+\gamma} + d^{7-p/2} \log^{3-p/2} d + \phi(O(d^{2+p} \log(1/\eps)/\eps^2), d)\right).$$

\begin{remark}
It is interesting to note that the work done by the server $C$ is just $\poly(d)$, while the majority of the work at Step $2, 4, 6, 8$, which costs $O(\nnz(\bar{M}) \cdot \log n)$ time, is done by the $k$ machines. This feature makes the algorithm fully scalable. 
\end{remark}